\documentclass[10pt,final,twocolumn]{IEEEtran}
\usepackage{amsthm,hyperref}
\usepackage{mathrsfs}
\usepackage{supertabular}
\usepackage{cite}
\usepackage{color}
\usepackage{paralist}
\usepackage{graphics} 
\usepackage{latexsym,  amsfonts, amssymb,graphicx,array,tabularx}
\usepackage{subfigure}
\usepackage{caption}
\usepackage{dsfont}
\usepackage{bbm}
\usepackage{epsfig}
\usepackage{amsmath}
\usepackage[english]{babel}
\usepackage{multicol}
\usepackage{epic}
\usepackage{epstopdf}
\usepackage{mathtools}

\theoremstyle{theorem}

\newtheorem{theorem}{Theorem}
\newtheorem{lemma}{Lemma}

\newtheorem{remark}{Remark}

\newtheorem{assumption}{Assumption}
\newtheorem{problem}{Problem}

\makeatletter

\newcommand{\Rmnum}[1]{\expandafter\@slowromancap\romannumeral #1@}
\makeatother

\newcommand\addtag{\refstepcounter{equation}\tag{\theequation}}
\DeclareMathOperator*{\esssup}{ess\,sup}
\DeclareMathOperator*{\argmax}{arg\,max}
\DeclareMathOperator*{\minimize}{minimize}
\DeclareMathOperator*{\subjectto}{subject \quad to}

\newcommand{\smallO}[1]{o{\left(#1\right)}}
\newcommand{\defeq}{\vcentcolon=}

\begin{document}
\IEEEoverridecommandlockouts

\title{Quickest Change Detection \\in Adaptive Censoring Sensor Networks}

\author{Xiaoqiang $\text{Ren}^{\star}$,  Karl H. $\text{Johansson}^{\dag}$, Dawei $\text{Shi}^{\ddag}$ and Ling $\text{Shi}^{\star}$

\thanks{The work by X. Ren and L. Shi was supported by a HK RGC theme-based project T23-701/14N.
The work by K. H. Johansson was supported in part by the Knut and Alice Wallenberg Foundation, the Swedish Foundation for Strategic Research,  and the Swedish Research Council.
The work by D. Shi was supported by the National Natural Science Foundation of China under Grant 61503027.

\indent $\star$: Department of Electronic and Computer Engineering,
Hong Kong University of Science and Technology, Kowloon, Hong Kong. Emails: {xren,~eesling}@ust.hk.

\indent $\dag$: ACCESS Linnaeus Center, School of Electrical Engineering, Royal Institute of Technology, Stockholm, Sweden. Email: kallej@ee.kth.se.

\indent $\ddag$: State Key Laboratory of Intelligent Control and Decision of Complex Systems; School of Automation, Beijing Institute of Technology, Beijing, 100081, China. Email: daweishi@bit.edu.cn.
}}
\maketitle

\begin{abstract}
The problem of quickest change detection with communication rate constraints is studied. A network of wireless sensors with limited computation capability monitors the environment and sends observations to a fusion center via wireless channels. At an unknown time instant, the distributions of observations at all the sensor nodes change simultaneously. Due to limited energy, the sensors cannot transmit at all the time instants. The objective is to detect the change at the fusion center as quickly as possible, subject to constraints on false detection and average communication rate between the sensors and the fusion center. A minimax formulation is proposed. The cumulative sum (CuSum) algorithm is used at the fusion center and censoring strategies are used at the sensor nodes. The censoring strategies, which are adaptive to the CuSum statistic, are fed back by the fusion center. The sensors only send observations that fall into prescribed sets to the fusion center. This CuSum adaptive censoring (CuSum-AC) algorithm is proved to be an equalizer rule and to be globally asymptotically optimal for any positive communication rate constraint, as the average run length to false alarm goes to infinity. It is also shown, by numerical examples, that the CuSum-AC algorithm provides a suitable trade-off between the detection performance and the communication rate.


\end{abstract}

\section{Introduction}
\emph{Background and Motivations:} The goal of quickest change detection is to detect the abrupt change in stochastic processes as quickly as possible subject to certain constraints on false detection. This problem has a wide range of applications, such as habitat monitoring~\cite{mainwaring2002wireless}, quality control engineering \cite{lai1995sequential}, computer security \cite{tartakovsky2006novel} and cognitive radio networks \cite{lai2008quickest}.
In the classical quickest change detection formulation, the decision maker observes a sequence of observations $\{X_1, \ldots, X_k,\ldots\}$, the distribution of which changes at an unknown time instant $\nu$. The observations before the change $X_{1},\ldots,X_{\nu-1}$ are independent and identically distributed (i.i.d.), and the observations after the change $X_{\nu}, \ldots, X_{\infty}$ are also i.i.d. but with a different distribution.
The change event model distinguishes two problem formulations: the Bayesian formulation due to Shiryaev \cite{Shiryaev1963,shiryaev2007optimal} and the minimax formulation due to Lorden \cite{lorden1971procedures} and Pollak \cite{pollak1985optimal}.



The classical quickest change detection problem does not consider the cost of acquiring observations. It assumes that the decision maker can access observations at all the time instants freely. This is an issue for resource-limited applications, such as those using wireless sensor networks (WSNs). In the problem of quickest change detection with WSNs, observations are taken by one or multiple sensors, which communicate with the decision maker via wireless channels \cite{ Veeravalli2001, mei2011quickest,Geng2013}. The limited resources, which include limited energy for each battery-powered sensor and the limited communication bandwidth, naturally pose the constraint that the observations cannot be sent to the decision maker continuously. Thus, we consider the problem of quickest change detection with such constraints.

\emph{Related Literatures and Contributions:} Recently, there are several works on constrained quickest change detection with minimax formulations~\cite{banerjee2011generalized,banerjee2012data, Geng2013, banerjee2014data}. Two classes of characterizations of the cost acquiring observations were considered: the cost of sampling~\cite{banerjee2012data, Geng2013, banerjee2014data} and the cost of communication~\cite{mei2011quickest,banerjee2011generalized,banerjee2014data}. In these works, algorithms consisting of stopping times and sampling/transmission schedulers were proposed. By transmission (sampling) schedulers, when local sensors send their data to a fusion center (a decision maker samples) is determined. The communication constraint was studied from the perspective of quantization as well as the communication rate in~\cite{yilmaz2012cooperative} but with general sequential detection settings, where an interesting point was that only one bit of information was sent to the fusion center whenever a transmission occured.

In this paper, we only take the cost of communication into account. This is motivated by WSNs applications for which the energy cost of sampling is usually negligible compared with that of communication~\cite{akyildiz2002wireless,dargie2010fundamentals}. Thus, it is a reasonable assumption that the sensors can take observations at each time instant but with limited number of communications with the fusion center. Furthermore, as in~\cite{mei2011quickest,banerjee2011generalized,banerjee2014data}, we do not consider quantization errors of the data sent from local sensors to the fusion center.


The structure of the system considered in this paper is illustrated in Fig. \ref{Fig:blockdiagramnetwork}.  Observations are taken by $M$ sensors and are sent to a fusion center via wireless channels. Due to limited energy, the remote sensors cannot transmit at all the time instants. To make the best use of the limited resources, the sensors are assumed to adopt censoring strategies \cite{rago1996censoring}. Each of the sensors samples at each time instant, but only transmits informative samples. The censoring strategies are adaptive to the detection statistic available at the fusion center. When necessary, the fusion center tells the sensors about the censoring strategies to use via the feedback channels\footnote{The feedback transmissions may cause additional energy consumption at the sensor side. But one should note that in our algorithm, the feedback message is quite simple (see Remark~\ref{Remark:ComputationComple}) and the feedback transmissions are usually quite few (see Remark~\ref{Remark:feedbackTimes}). In particular, when CuSum-AC algorithm with $N=2$ levels is used, only one bit of information is needed when feedback occurs.}.

To deal with communication constraints, in the existing literature~\cite{mei2011quickest,banerjee2011generalized,banerjee2014data,banerjee2013data1}, CuSum-like algorithms (CuSum in~\cite{mei2011quickest,banerjee2011generalized}, a variant called DE-CuSum in~\cite{banerjee2014data,banerjee2013data1}) are run locally at the sensor nodes and \emph{the detection statistic of the algorithm} is sent to the fusion center only when it is above a certain threshold.
In this paper, a fundamentally different approach is adopted: \emph{the observations} instead of the detection statistic are censored and transmitted to the fusion center.  Compared to running CuSum-like algorithms, the online computation load required at the sensor side to censor the observations is reduced (see Remark~\ref{Remark:ComputationComple}). Another advantage is that our algorithm is an equalizer rule, which helps reduce off-line computation complexity (see Remark~\ref{Remark:EqualizerRule}).
Compared with decentralized settings in ~\cite{mei2011quickest,banerjee2011generalized,banerjee2014data,banerjee2013data1}, where the communication is unidirectional and remote sensors implement the censoring strategy in an
autonomous manner, our algorithm is centralized in the sense that the censoring strategies used at remote sensors are fed back by the fusion center. Evidently, the feedback transmission introduced complicates the system, although it occurs only occasionally and its message is rather simple. Our algorithm, however, is able to reduce system complexity by reducing the number of sensors required. More specifically, to achieve the same detection performance, fewer sensors are required compared with the decentralized counterparts. A similar idea can be found in~\cite{huang2006dynamic}. About the decentralized censoring strategies, we point out that to utilize the information of past observations available at each sensor, it is necessary to censor the detection statistic instead of current observations. Note that in this paper, we do not consider the cost of sampling, while if the sampling is ``energy consuming" (i.e., the energy consumed by sampling is comparable to that by communication), then the DE-CuSum algorithm running locally at the sensor nodes~\cite{banerjee2014data}, which skips sampling when necessary, is desirable.

In summary, the main contribution is that for the quickest change detection with communication rate constraints in the minimax formulation, a novel algorithm is proposed, which is the CuSum algorithm coupled with adaptive censoring strategies (CuSum-AC). The CuSum-AC algorithm is proven to be asymptotically optimal for any positive communication rate constraint and thus provides some insights into how one can utilize censoring strategies \emph{adaptively} to achieve the globally asymptotic optimality.


\setlength{\unitlength}{1.25mm}
\begin{figure}[tb]
\thicklines
\centering
\begin{picture}(70,23)(-3,-12)
\thicklines

{\scriptsize

\put(0,-9){\framebox(11,18)}

\put(0.5,0.2){${\rm Monitored}$}
\put(2.0,-1.5){${\rm target}$}

\put(11,7){\vector(1,0){7}}
\put(18,4){\framebox(10,6)}
\put(19.3,6.5){${\rm Sensor}\:1$}

\put(28,8){\line(1,0){14}}
\put(30.5,9.5){$\gamma_{\{1,k\}}X_{\{1,k\}}$}
\dashline{0.98}(28.1,6.5)(40.5,6.5)
\put(28.1,6.5){\vector(-1,0){0.1}}
\put(32,4.5){$\psi_{\{1,k\}}$}

\dashline{0.9}(40.5,6.5)(40.5,2)
\put(42,8){\line(0,-1){4.5}}
\put(42,3.5){\vector(1,0){5}}
\dashline{0.9}(40.5,2)(47,2)

\put(11,-7){\vector(1,0){7}}
\put(18,-10){\framebox(10,6)}
\put(18.3,-7.5){${\rm Sensor}\:M$}

\put(28,-8){\line(1,0){14}}
\put(30.5,-10){$\gamma_{\{M,k\}}X_{\{M,k\}}$}
\dashline{0.98}(28.1,-6.5)(40.5,-6.5)
\put(28.1,-6.5){\vector(-1,0){0.1}}
\put(32,-5.5){$\psi_{\{M,k\}}$}

\dashline{0.9}(40.5,-6.5)(40.5,-2)
\put(42,-8){\line(0,1){4.5}}
\put(42,-3.5){\vector(1,0){5}}
\dashline{0.9}(40.5,-2)(47,-2)

\put(47,-5){\framebox(7,10)}
\put(47.4,0.5){${\rm Fusion}$}
\put(47.6,-2){${\rm center}$}

\put(54,0){\vector(1,0){10}}
\put(55,1){${\rm Change?}$}

{\Large
\put(23,-1.5){$.$}
\put(23,0){$.$}
\put(23,1.5){$.$}
}

}


%
%
%
%
%
%
%
%
%
%
%
%
%
%

\end{picture}
 \caption{ The quickest change detection system in adaptive censoring sensor networks. Each sensor corresponds to the blue rectangle in Fig.~\ref{Fig:blockdiagram}.} \label{Fig:blockdiagramnetwork}
\end{figure}
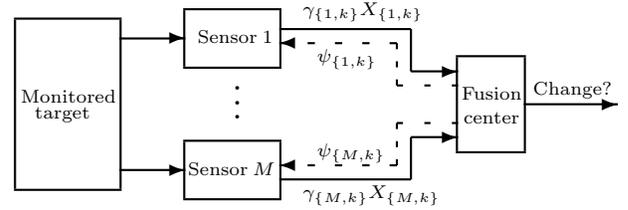

\emph{Organizations:} The remainder of this paper is organized as follows. The one sensor case is studied in Sections \ref{Section:problem setup}-\ref{Section:MainResults}. The mathematical formulation of the considered problem is given in Section \ref{Section:problem setup}. We present the CuSum-AC algorithm in Section \ref{Section:Algorithm}. The main results are given in Section \ref{Section:MainResults}. First we show that the CuSum-AC algorithm is an equalizer rule, i.e., the worst-worst case detection delay is attained whenever the change event happens. Then we prove that the CuSum-AC algorithm is globally asymptotically optimal for any positive communication rate constraint. We generalize the results obtained for the one sensor case to the multiple sensors scenario in Section \ref{Section:ExtensiontoMultipleSensors}. Numerical examples are given in Section \ref{Section:Numerical Example} to illustrate the main results. Some concluding remarks are given in the end. All the proofs are presented in Appendices.

\textit{Notations}: $\mathbb{N}$, $\mathbb{N}_{+}$, $\mathbb{R}$, $\mathbb{R}_{+}$ and $\mathbb{R}_{++}$ are the set of non-negative integers, positive integers, real numbers, non-negative real numbers and positive real numbers, respectively.
$k\in \mathbb{N}$ is the time index.
$\mathbf{1}_{A}$ represents the indicator function that takes value $1$ on the set $A$ and $0$ otherwise.
$\times$ stands for the Cartesian product. For $x\in\mathbb{R}$, $(x)^{+}=\max(0,x).$

\section{Problem Setup} \label{Section:problem setup}
For simplicity of presentation, first we consider the one sensor case; see Fig.~\ref{Fig:blockdiagram}. Then we extend the results to the sensor networks scenario in Section \ref{Section:ExtensiontoMultipleSensors}.
Consider the change detection system in Fig. \ref{Fig:blockdiagram}.
A sequence of observations, say $\{X_k\}_{k\in\mathbb{N}_{+}}$, about
the monitored environment are taken locally at the sensor.
Assume that $\nu$ is an unknown (but not random) time instant when a change event takes place. The instant may be $\infty$, corresponding to that the change never happens.
The observations at the sensor before $\nu$, $\{X_1,\ldots,X_{\nu-1}\}$, are i.i.d. with probability density function (pdf) $f_0$, and the observations from $\nu$ on, $X_{\nu},X_{\nu+1},\ldots$, are i.i.d. with pdf $f_1$.
Let $\mathbb{P}_{\nu}$ denote the probability measure when the change happens at $\nu$. If there is no change, we denote this measure by $\mathbb{P}_{\infty}$. The expectation $\mathbb{E}_{\nu}$ and $\mathbb{E}_{\infty}$ are defined accordingly.

\setlength{\unitlength}{1.25mm}
\begin{figure}[tp]
\thicklines
\centering
\begin{picture}
(75,20)(0,-12)
\thicklines

{\tiny
\put(0,2.5){${\rm Monitored}$}
\put(2,0.5){${\rm target}$}
\put(0,0){\vector(1,0){13}}

{\color{blue}\put(10,-10){\framebox(33,17)}}

{\footnotesize
\put(10.5,5){${\rm Sensor}$}
}

\put(13,-2){\framebox(12,4)}

\put(13.5,0.2){${\rm Measurement}$}
\put(14.3,-1.7){${\rm acquisition}$}
\put(26.5,1.5){$X_k$}

\put(25,0){\line(1,0){9}}

\put(30,-9){\framebox(10,3)}
\put(31.1,-8){${\rm Censoring}$}
\put(36,-4){$\gamma_k$}

\put(34,0){\line(1,1){3}}
\put(37,0){\vector(1,0){16}}
\put(53,-5){\framebox(7.5,10)}

\put(53.2,0.2){${\rm Decision}$}
\put(54.3,-1.8){${\rm maker}$}

\dashline{0.98}(45,-7.5)(40,-7.5)
\put(40,-7.5){\vector(-1,0){0.1}}
\dashline{0.91}(45,-7.5)(45,-3)
\dashline{0.91}(45,-3)(53,-3)
\put(47,-5){$\psi_k$}
\put(46,1.5){$\gamma_kX_k$}

\put(60.5,0){\vector(1,0){9}}

\put(61,1){${\rm Change?}$}

\put(35.5,-6){\vector(0,1){6.5}}

\put(27,-7.5){\vector(1,0){3}}
\put(27,0){\line(0,-1){7.5}}
}
\end{picture}
\caption{ The quickest change detection system with a sensor that adopts an adaptive censoring strategy.} \label{Fig:blockdiagram}
\end{figure}
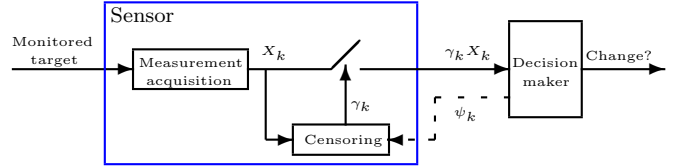

To characterize the behavior that the sensor cannot send the observation $X_k$ to the decision maker all the time, we introduce a binary variable $\gamma_k$ as
\begin{align} \label{Equation:IndicatorOneSensor}
    \gamma_k = \left\{
        \begin{array}{ll}
            1, & \text{if $X_k$ is sent to the decision maker},\\
            0, & \text{otherwise}.
        \end{array}
    \right.
\end{align}
Then the information pattern available for the decision maker at the time instant $k$ is given by
$\mathcal{I}_k=\big\{(1,\gamma_1,\gamma_1X_1),\ldots,(k,\gamma_k,\gamma_kX_k)\big\}$
with $\mathcal{I}_0=\{\emptyset\}$. A random variable $T\in\mathbb{N}_{+}$ is called a stopping time if $\{T=k\}\in\sigma(\mathcal{I}_{k})$, where $\sigma(\mathcal{I}_{k})$ is the smallest sigma-algebra of $\mathcal{I}_{k}$. A stopping time can be characterised by a stopping rule, which is a mechanism that decides whether or not to stop based on the available information.

To make the best use of the limited communication resources, the censoring strategy is implemented at the sensor node. We consider an adaptive censoring strategy, which varies with the information pattern. Specifically, the censoring strategy used at the sensor node at time instant $k$, which is denoted by $\psi_k$, is determined by the decision maker based on $\mathcal{I}_{k-1}$. When $\psi_k \neq \psi_{k-1}$, the decision maker sends $\psi_k$ to the sensor through the feedback channel. Since the sensor is assumed to have no memory and can thus only access $X_k$ at time $k$, the censoring strategy $\psi_k: \mathbb{R} \mapsto \{0,1\}$ has the form as
$\gamma_k=\psi_k(X_k).$
The censoring policy along the horizon is given by
$\Psi=(\psi_1,\ldots,\psi_{T}).$

The communication constraint is formulated as the limited communication rate \emph{before} the change event happens. It depends on the censoring policy $\Psi$ and is formalized as
\begin{align}
    r(\Psi)=\limsup_{n\to\infty}\frac{1}{n}\mathbb{E}_{\infty}\left[\sum_{k=1}^{n}\gamma_k \big | T \geq n\right]\leq \epsilon, \label{Equation:CommunicationConstraint}
\end{align}
where $0<\epsilon\leq 1$ is a design parameter.
By adjusting $\epsilon$, a tradeoff between communication resources and detection performance is obtained. Note that the post-change period (i.e., the detection delay) is usually quite small compared with the pre-change period, hence we only pose the communication constraint \emph{before} the change. The conditional expectation $\mathbb{E}_{\infty}[\cdot|T\geq n]$ thus is considered.  The asymptotic optimality result
of the paper does not hold if the total cost is considered (see Remark~\ref{Remark:TotalCost}).
A similar criterion called pre-change transmission cost is considered in \cite{banerjee2013data1}.

For the detection performance of the quickest change detection, there are two indices: the risk of false detection and the detection delay. Given $T$ and $\Psi$, the risk of false detection is characterized by the average run length to false alarm (ARLFA)
\begin{align*}
    g(T,\Psi)=\mathbb{E}_{\infty}\left[ T \right];
\end{align*}
cf., \cite{lorden1971procedures,tartakovsky2012third}.
Note that the reciprocal of the ARLFA is connected to the false alarm rate. When the ARLFA goes to infinity, the false alarm rate goes to zero.
The stopping time $T$ is related to $\mathcal{I}_k$, which is determined by the observation sequence $\{X_k\}$ and the censoring policy $\Psi$, so the ARLFA is related with $\Psi$. To highlight this dependence, we use $g(T,\Psi)$ in the above definition.

For the detection delay, we consider Lorden's worst-worst case detection delay~\cite{lorden1971procedures}\footnote{It is easy to see that the main result in this paper, i.e., asymptotical optimality of our algorithm, also holds when Pollak's criterion~\cite{pollak1985optimal} is considered.}, which is given by
\begin{align}
    d_{L}(T,\Psi) = \sup_{1\leq\nu<\infty} \big\{ \esssup_{\mathcal{I}_{\nu-1}} \mathbb{E}_{\nu}[(T-\nu+1)^+|\mathcal{I}_{\nu-1}] \big\}.
    \label{Equation:DetectionDelayLorden}
\end{align}
\begin{problem}  \label{Problem:Lorden}
    \begin{align*}
        \minimize_{T,\Psi}& \qquad d_L(T,\Psi),\\
        \subjectto& \qquad g(T,\Psi) \geq \zeta,   \addtag \label{Equantion:ProCons1}         \\
                  & \qquad r(\Psi)\leq \epsilon,    \addtag \label{Equantion:ProCons2}
    \end{align*}
where $\zeta\geq1$ is a given lower bound of the ARLFA.
\end{problem}
Note that for the classical formulation of the quickest change detection, the observations are assumed to be i.i.d. conditioned on the change event.
While since $\Psi$ is adaptive, the available observation sequence $\{\gamma_k, \gamma_kX_k\}$ are correlated across the time. To solve the above problems thus is quite challenging.

To avoid degenerate problems, we make the following assumption for the remainder of this paper.
\begin{assumption}  \label{Assumption:FiniteKLDivergence}
\begin{align*}
0<\mathbf{I}(f_{1}||f_0)<\infty, \quad
0<\mathbf{I}(f_0||f_1)<\infty,
\end{align*}
where
$\mathbf{I}(f_{1}||f_0) = \int_{\mathbb{R}} f_1(x)\ln\frac{f_1(x)}{f_0(x)} \mathrm{d}x$,
$\mathbf{I}(f_{0}||f_1) = \int_{\mathbb{R}} f_0(x)\ln\frac{f_0(x)}{f_1(x)} \mathrm{d}x$
are the Kullback--Leibler (K--L) divergences.
\end{assumption}



%
%
%
%

Our subsequent analysis utilizes the CuSum algorithm, which is stated as follows. Let constant $a$ be a given threshold and $\ell(X_k)=\ln\frac{f_1(X_k)}{f_0(X_k)}$ the log-likelihood ratio function.
The stopping time for the CuSum algorithm thus is computed as
\begin{align}
    T(a)=\inf\{k:c_k > a\}, \label{Equation:CuSumStoppingTime}
\end{align}
where $c_k$ is the detection statistic for the CuSum algorithm computed by
\[c_k = (c_{k-1} + \ell(X_k))^+\]
with $c_0 = 0$.
The CuSum algorithm is optimal for original Lorden's formulation when $a$ is chosen such that
$\mathbb{E}_{\infty}\left[ T(a) \right]=\zeta$~\cite{moustakides1986optimal, ritov1990decision}.
When there is no communication rate constraint, i.e., $\epsilon=1$, Problem~\ref{Problem:Lorden} is reduced to original Lorden's formulation.
We should remark that it is difficult to find strictly optimal algorithms for Problem \ref{Problem:Lorden}. We hence focus on asymptotically optimal solution. For simplicity, we use $T(a)$ to denote the CuSum algorithm with the threshold $a$ for the remainder of this paper.

\section{CuSum-AC Algorithm} \label{Section:Algorithm}
In this section, we present the proposed CuSum-AC algorithm, which is the CuSum algorithm coupled with a censoring policy that adaptively switches between different censoring strategies.
We say a CuSum-AC algorithm is with $N$ levels if the number of censoring strategies used is $N$.
For ease of presentation, we present the CuSum-AC algorithm with $N=2$ levels, and for the CuSum algorithm with $N>2$ levels, see~\eqref{Equation:NlevelsCensoring} and Remark~\ref{Remark:Nlevels}. In the remainder of this paper, if not particularly indicated, the CuSum-AC algorithm refers to the one with $N=2$ levels. The algorithm consists of three parts: how the detection statistic updates, what the adaptive censoring policy is, and when the algorithm stops and declares the change.

Let $a$ and $a_1$ be given thresholds with $a_1 < a$. The detection statistic $s_k$ are updated as follows:
\begin{align}
    \tilde{s}_k=&\max\{0,s_{k-1}+\ell^{\psi_k}(\gamma_k,X_k)\},\label{Equation:DetecStatisUpdateNLevels0}\\
s_k = &\left\{
        \begin{array}{ll}
             a_1, & \text{if}\quad s_{k-1}<a_1\:\text{and}\: \tilde{s}_k\geq a_1,\\
            \tilde{s}_k, & \text{otherwise},
        \end{array}
    \right. \label{Equation:2LevelsUpdate}
\end{align}
with initial value $s_0=0$. The quantity $\ell^{\psi_k}(\gamma_k,X_k)$ is the log-likelihood ratio function of the random variable $\gamma_kX_k$\footnote{Note that since a  censoring strategy is adopted, when $\gamma_k=0$, the decision maker still has a rough information about $X_k$.} under the censoring strategy $\psi_k$:
\begin{align*}
\ell^{\psi_k}(\gamma_k,X_k) &= \ln \frac{f_1(X_k)}{f_0(X_k)} (\gamma_k-0)\\
 &\qquad +\, \ln \frac{\mathbb{P}_{1}\{\gamma_k=0|\psi_k\}}{\mathbb{P}_{\infty}\{\gamma_k=0|\psi_k\}} (1-\gamma_k)
\end{align*}

The adaptive censoring policy is given by
\begin{align} \label{Equation:CensoringStrategyNlevels}
\psi_k = &\left\{
        \begin{array}{ll}
             \psi^{*}(1), & \text{if}\quad s_{k-1} \geq a_1,\\
            \psi^{*}(\epsilon_1), & \text{if} \quad s_{k-1} < a_1,
        \end{array}
    \right.
\end{align}
where $0<\epsilon_1 \leq 1$ and $\psi^{*}(\epsilon_1)$ is defined as follows.
Let $0<e\leq 1$, define
\begin{align} \label{Equation:StationaryCensoringPolicy}
    \psi^{*}(e)=\argmax_{\psi \in \mathcal{C}(e)} \mathbf{I}^{\psi}(f_1||f_0),
\end{align}
where
$\mathcal{C}(\epsilon)=\left\{\psi: \mathbb{P}_{\infty}\left\{\psi\left(X_k\right)=1\right\}=\epsilon\right\}$, and $\mathbf{I}^{\psi}(f_1||f_0)$ is the K--L divergence of the observations available at the decision maker under the censoring strategy $\psi$:
\begin{align*}
\mathbf{I}^{\psi}(f_1||f_0) &=\mathbb{E}_{1}\left[ \ell^{\psi}(\gamma_1,X_1)\right].
\end{align*}
Among the censoring strategies that have communication rate~$\epsilon$, the strategy $\psi^{*}(\epsilon)$ has the maximal post-censoring K--L divergence. In general, $\psi^{*}(\epsilon)$ does not have analytic expressions, but it is well known that $\psi^{*}(\epsilon)$ has a special structure: the likelihood ratio of the no-send region is \emph{a single interval}~\cite{rago1996censoring}. The upper and lower bounds of this single interval is obtained via numerical simulations.

The stopping time for the CuSum-AC algorithm is given by
\begin{align}  \label{Equation:StopTimeAC}
    T_{c}(a)=\inf\{k:s_k\geq a\}.
\end{align}
In the following, we use $T_c(x,y)$ to denote the stopping time as $T_c(a)$ for which the initial statistic is $s_0 = x$ and the threshold is $y$.

In summary, the CuSum-AC algorithm is illustrated in Fig.~\ref{Fig:CuSum-AC}. A few remarks on the algorithm are presented as follows.

Compared with the CuSum algorithm, two additional parameters $\epsilon_1, a_1$ are introduced for the CuSum-AC algorithm. As in the CuSum algorithm, the parameters involved in the CuSum-AC algorithm, i.e., $\epsilon_1, a_1$ together with $a$, can only be determined by numerical simulations. In practice, if the required ARLFA is large enough (i.e., the threshold $a$ is chosen large enough), the communication rate constraint and the ARLFA constraint can be met independently. More specifically, to meet the communication rate constraint, one may first fix a large enough $a$ (any value) and then obtain an appropriate pair $(a_1,\epsilon_1)$. Since given a pair $(\epsilon_1,a_1)$, when $a$ is large enough, the communication rate is almost a constant, one then may meet the ARLFA constraint by just adjusting $a$.

The detection statistic of the CuSum-AC algorithm is reset to switching threshold $a_1$ whenever it crosses $a_1$ from below. This facilitates the asymptotic optimality analysis of the CuSum-AC algorithm and makes the stopping time $T_{c}(a)$ of the CuSum-AC algorithm  an equalizer rule (the details of which are given in Section~\ref{Section:MainResults}).

We now elaborate on the adaptive censoring strategy. Note that by definition, for a fixed $0<e\leq 1$, $\psi^*(e)$ is the most informative in the sense that it achieves the largest post-censoring K--L divergence with communication rate $e$. While by adjusting $e$, one can trade off communication cost against information quality of $\psi^*(e)$. On the one hand, the larger $e$ is, the more information of the observations taken at the sensor node is conveyed to the decision maker under $\psi^*(e)$. On the other hand, if $e_1 > e_2$, $\psi^*(e_1)$ incurs more communication cost than $\psi^*(e_2)$ does. Intuitively, one tends to use a censoring strategy $\psi^*(e)$ with a larger $e$ when it is deemed ``more important". At each time $k$, our adaptive censoring strategy in~\eqref{Equation:CensoringStrategyNlevels} tends to use a censoring strategy $\psi^*(e)$ with a larger $e$ when $s_{k-1}$ is larger. This idea comes from the observation of the typical evolution of the CuSum algorithm as illustrated in Fig.~\ref{Fig:TimeSeriesCuSum}. The detection statistic $c_k$ goes up and down before it reaches the threshold. At most times, the detection statistic stays small.
Note that if the sojourn time when $c_k$ stays in one interval is large enough, the change of $c_k$ in that interval can be approximated using the \emph{statistical property} of the observations without knowing each observation. Let us take two extreme cases for example. Let $T_1$ and $T_2$ be the sojourn time when $c_k$ is in interval $1$ and interval $2$, respectively.
Suppose that $T_1$ is sufficiently large, while $T_2 = 1$.
Then by the renewal theorem~\cite{woodroofe1982nonlinear}, one knows that the change of $c_k$ in interval $1$ can be obtained by $\Delta_1(c_k)\approx T_1\mathbf{I}(f_1||f_0)$. This means that almost no information is lost for the decision maker even if no messages are sent by the sensor. However, for the case $T_2=1$, in order to make the change $\Delta_2(c_k)$ known to the decision maker, this single observation has to be sent to the decision maker, the communication rate of which is $1$.
Based on this notion, one extends the adaptive censoring strategy to $N>2$ levels as follows:
\begin{align} \label{Equation:NlevelsCensoring}
\psi_k = &\left\{
        \begin{array}{ll}
             \psi^{*}(1), & \text{if}\quad s_{k-1} \geq a_1,\\
            \psi^{*}(\epsilon_1), & \text{if}\quad a_2 \leq s_{k-1} < a_1,\\
            \quad\vdots &\qquad\qquad \vdots \\
            \psi^{*}(\epsilon_{N-2}), & \text{if}\quad a_{N-1} \leq s_{k-1} < a_{N-2},\\
            \psi^{*}(\epsilon_{N-1}), & \text{if}\quad s_{k-1} < a_{N-1},
        \end{array}
    \right.
\end{align}
with $0<\epsilon_{N-1} \leq \epsilon_{N-2} \leq \cdots \leq \epsilon_1 \leq 1 $ and $a_{N-1} < a_{N-2} < \cdots<a_1$.

\setlength{\unitlength}{1.25mm}
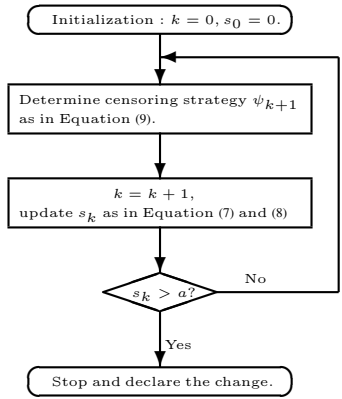
\begin{figure}[tp]
\thicklines
\centering
\begin{picture}
(40,43)(-20,-33)
\thicklines

{\tiny
\put(0,8){\oval(28,3)}
\put(-11.5,7.5){${\rm Initialization}: k=0, s_0=0.$ }

\put(0,6.5){\vector(0,-1){5.5}}
\put(-16,-4){\framebox(32,5)}
\put(0,-4){\vector(0,-1){5}}
\put(-15,-1){${\rm Determine\:censoring\:strategy}$ $\psi_{k+1}$}
\put(-15,-3){${\rm as\:in\:Equation~\eqref{Equation:CensoringStrategyNlevels}.}$}

\put(-16,-14){\framebox(32,5)}
\put(-5,-11){$k=k+1,$ }
\put(-15,-13){${\rm update}\:s_k\: {\rm as\:in\:Equation~\eqref{Equation:DetecStatisUpdateNLevels0}\:and\:\eqref{Equation:2LevelsUpdate} }$ }
\put(0,-14){\vector(0,-1){5}}

\put(-6,-21){\line(3,1){5.9}}
\put(6,-21){\line(-3,1){5.9}}
\put(-6,-21){\line(3,-1){5.9}}
\put(6,-21){\line(-3,-1){5.9}}
\put(-3,-21.5){$s_k>a?$ }
\put(6,-21){\line(1,0){13}}
\put(19,-21){\line(0,1){25}}
\put(19,4){\vector(-1,0){19}}
\put(9,-20){${\rm No}$}

\put(0,-23){\vector(0,-1){6}}
\put(0.5,-27){${\rm Yes}$}
\put(0,-30.5){\oval(28,3)}
\put(-11.5,-31){${\rm Stop\:and\:declare\: the\: change.}$ }
}
\end{picture}
 \caption{ CuSum-AC algorithm } \label{Fig:CuSum-AC}
\end{figure}

\begin{remark} \label{Remark:ComputationComple}
Now we discuss the practical implementation of the CuSum-AC algorithm with general $N$ levels and the online computation load at the sensor side. The parameters, i.e., $\psi^*(\epsilon_{N-1}),\ldots,\psi^*(\epsilon_{1}), a_{N-1},\ldots,a_1$ and $a$, are determined prior to the system run time. The censoring strategies $\psi^*(\epsilon_n), 1\leq n \leq N-1$ are stored in the sensor node\footnote{There is no need to store $\psi^*(1)$, under which no observations are censored at all. } and the feedback message from the fusion center is the strategy index~$n$ (together with $n=0$ representing $\psi^*(1)$). The feedback happens when $\psi_k \neq \psi_{k-1}$. Note that $\psi^*(\epsilon_n)$ has a special structure: the likelihood ratio of the no-send region is a single interval.
Hence, to store a censoring strategy $\psi^*(\epsilon_n)$, it suffices to store the corresponding lower and upper bounds of the likelihood ratio (or the observations in some special cases, see below).
The only computation task of the remote sensor is to implement the censoring strategies $\psi^*(\epsilon_n)$, the computational load of which is explained for the following two cases.
For general distributions $f_1$ and $f_0$, the sensor
first computes the likelihood ratio of $X_k$ and then compares it to the upper and lower bounds. Note that to run the CuSum-like algorithms locally, the remote sensor needs to further compute the logarithm of the likelihood ratio. Since comparisons have negligible computational load compared with that of computing a logarithm, the computational load at the remote senor for our algorithm is much lower.
If the distributions $f_1$ and $f_0$ are such that the likelihood ratio function is monotone, then a single interval of the likelihood ratio also implies a single interval of the observations. Then to implement $\psi^*(\epsilon_n)$, the sensor just needs to compare the observations directly to the corresponding lower and upper bounds of the observations.
The family of distributions that have monotone likelihood ratio property is quite large, e.g.,  exponential, Binomial, Poisson and normal distributions with known variances. Although, in most of these cases, the computational load of log-likelihood ratio is also low-- one does not need to actually compute the logarithm but only elementary computations are required. Still, compared with running CuSum-like algorithms, the computational load at the remote sensor is considerably reduced in our algorithm, since comparisons are much simpler  to compute than multiplications, especially when the observations take on real values.
\end{remark}


\begin{figure}
  \centering
  \includegraphics[scale=0.4]{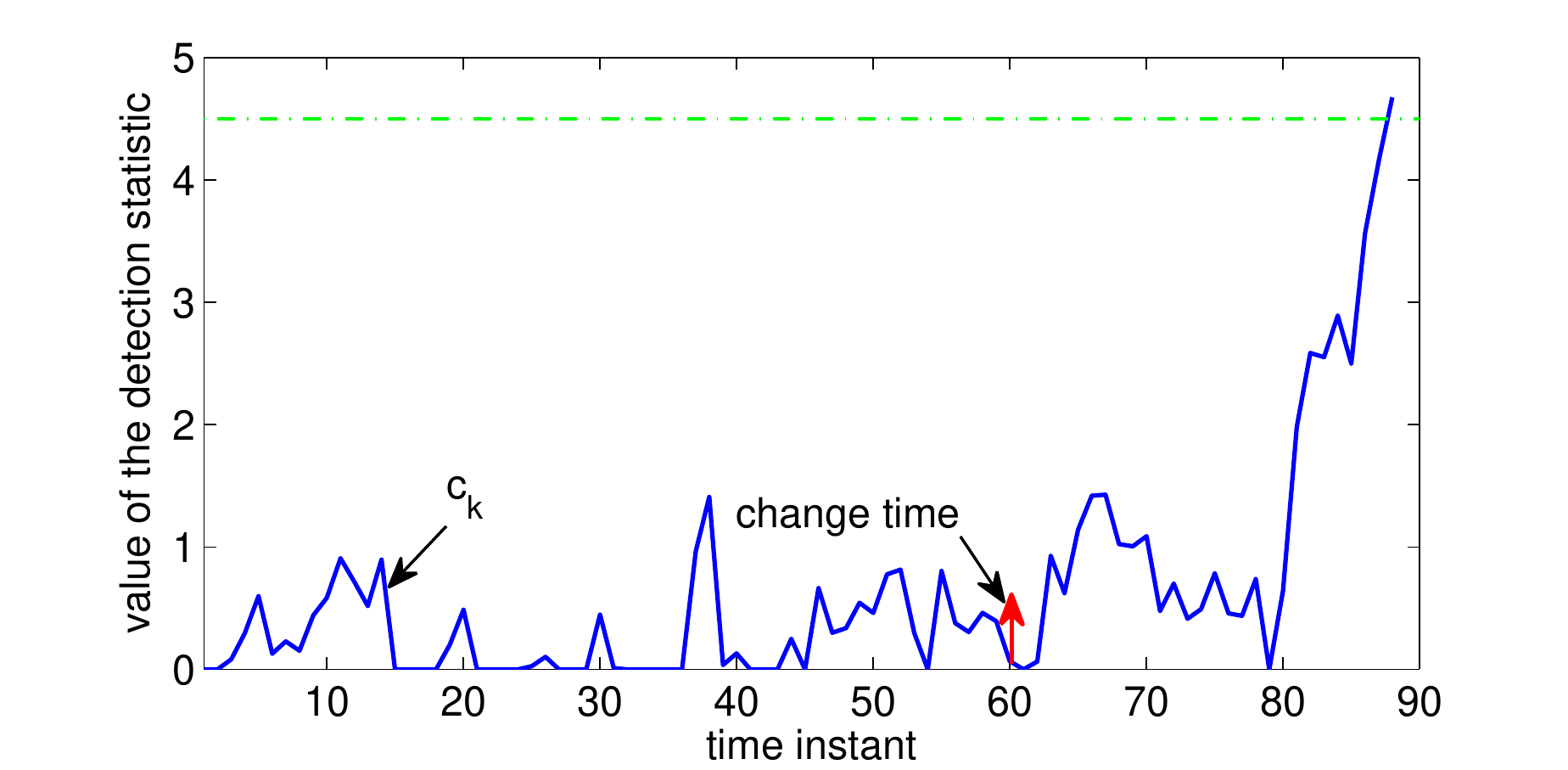}\vspace{-2mm}
  \caption{ Typical evolution of the CuSum algorithm. There is a mean shift in Gaussian noise, where the parameters used are as follows: before the change $\mathbb{P}_{\infty}: X_k \sim \mathcal{N}(0,1)$, after the change $\mathbb{P}_{1}: X_k \sim \mathcal{N}(0.5,1)$, change time $\nu=60$ and threshold $a=4.5$.} \label{Fig:TimeSeriesCuSum}
  \vspace{-4mm}
\end{figure}

\section{Performance Analysis} \label{Section:MainResults}
%

In this section, we first show that the $T_c(a)$ is an equalizer rule, i.e., the detection delay $d_{L}(T_c(a),\Psi)$ is attained for any change time $\nu$. We then prove $T_c(a)$ is asymptotically optimal for any communication rate constraint.

\subsection{Supporting Definitions} \label{section:supportingDefinition}
The classical performance analysis of the CuSum algorithm interprets the CuSum algorithm as a sequence of two-sided sequential probability ratio tests (SPRTs) \cite{poor2009quickest}.  This technique is also used for our analysis of the CuSum-AC algorithm. Intuitively, the CuSum-AC algorithm is a sequence of two-sided ($0$ and $a$) SPRTs with switching modes (original or censored) of observations. Each time the detection statistic crosses $a_1$ from below, it is reset to be $a_1$. This behavior is mathematically characterized as follows.

Define a stopping time of an SPRT with a starting point $0\leq z <a-a_1$ as a variable:
\begin{align*}
    \eta(z)= & \inf \left\{n: z+\sum_{k=1}^{n}\ell(X_k) \notin [0,a-a_1]\right\}.
\end{align*}
Note that $\eta(0)$ can be viewed as the first time that the detection statistic jumps out from $[a_1,a]$ with the initial point $a_1$. It either crosses the threshold $a$ or returns to $[0,a_1]$ and starts a test with censored observations. We denote by $\hat{s}_{\eta(z)}$ the detection statistic at the time instant $\eta(z)$ bounded below by zero:
\begin{align*}
\hat{s}_{\eta(z)}=\left(z+\sum_{k=1}^{\eta(z)}\ell(X_k)+a_1\right)^{+}.
\end{align*}

Define a detection statistic $\breve{s}_k(z)$, which is updated in the same manner with that in the CuSum algorithm but with an initial point $0\leq z < a_1$ and censored observations. The details are as follows:
\begin{align*}
  \breve{s}_k(z)&=\left( \breve{s}_{k-1}(z)+\ell^{\psi^{*}(\epsilon_1)}(\gamma_k,X_k) \right)^{+}, \addtag
  \label{Equation:CensoredCuSum}\\
  \breve{s}_0(z)&=z.
\end{align*}
Based on $\breve{s}_k(z)$, we define a stopping time by
\begin{align*}
\phi(z)=&\inf\left\{ k: \breve{s}_k(z) \geq a_1 \right\}.
\end{align*}
As the CuSum-AC algorithm starts at $0$, $\phi(0)$ can be interpreted as the first time that it reaches $a_1$ and switches the observation mode from the censored one to the original one.

Let
\begin{align}
\Phi = \eta(0)+ \phi\left(\hat{s}_{\eta(0)}\right)\mathbf{1}_{\{\hat{s}_{\eta(0)}<a_1\}}. \label{Equation:PhiDistribution}
\end{align}
The CuSum-AC algorithm can be interpreted as a sequence of SPRTs with stopping times being of two distributions. Specifically, the CuSum-AC algorithm starts with the stopping time distributed as $\phi(0)$, and after the time instant $\phi(0)$, it is a sequence of SPRTs with stopping times i.i.d. distributed as~$\Phi$.

\subsection{Equalizer Rule and Asymptotic Optimality}
\begin{theorem} \label{Theorem:Equalizer}
The stopping time $T_c(a)$ is an equalizer rule for Problem \ref{Problem:Lorden}, i.e.,
\begin{align*}
d_{L}(T_c(a),\Psi) = \esssup_{\mathcal{I}_{\nu-1}} \mathbb{E}_{\nu}[(T_c(a)-\nu+1)^+|\mathcal{I}_{\nu-1}], \: \forall \nu\geq 1.
\end{align*}
\end{theorem}
\begin{remark} \label{Remark:EqualizerRule}
In general, the parameters used for the algorithm (i.e., $\epsilon_1, a_1$ and $a$)
can only be obtained by numerically simulating the detection performance (i.e., the delay and the ARLFA).
The above theorem means that the change time does not affect the value of $d_{L}(T,\Psi)$. For simplicity, we can just let $\nu=1$ to simulate the delay.
\end{remark}


We now focus on asymptotic optimality of $T_c(a)$.
Before presenting the main theorem, we first present the supporting lemma about the communication rate of the CuSum-AC algorithm as follows.
\begin{lemma} \label{Proposition:ComRateSatisfied}
Given any finite $a_1>0$ and $0<\epsilon\leq 1$, there exists
a non-empty set $\mathcal{E}(a_1,\epsilon)$ such that
when $\psi^{*}(\epsilon_1)$ with $\epsilon_1\in\mathcal{E}(a_1,\epsilon)$ is used, the communication rate constraint is uniformly satisfied for any $a>a_1$ (including $+\infty$). In other words, given any finite $a_1>0$ and $0<\epsilon\leq 1$, there exists a censoring strategy $\Psi$ as in~\eqref{Equation:CensoringStrategyNlevels} with $\epsilon_1\in\mathcal{E}(a_1,\epsilon)$, such that
\begin{align*}
r(\Psi)\leq \epsilon, \quad \forall a\in(a_1,+\infty].
\end{align*}
\end{lemma}

The asymptotic optimality analysis involves the scenario where the threshold $a\to\infty$. The above theorem enables us to study the asymptotic performance of the CuSum-AC algorithm without worrying whether the communication constraint will be violated for some $a$.

%

Given $a_1$ and $\epsilon$, we define a set $\mathcal{E}^{*}(a_1,\epsilon)$ as
\begin{align}
\mathcal{E}^{*}(a_1,\epsilon)= \mathcal{E}(a_1,\epsilon) \cap \mathcal{E}'(a_1,\epsilon),   \label{Equation:OptimalEpsilonSetDefinition}
\end{align}
where the set $\mathcal{E}'(a_1,\epsilon)$ is given by
\begin{align*}
&\mathcal{E}'(a_1,\epsilon) \\
=& \{\epsilon_1>0: \mathbb{E}_{\infty}\left[ \phi\left(\hat{s}_{\eta(0)}\right) |\hat{s}_{\eta(0)} < a_1 \right] \geq \mathbb{E}_{\infty}\left[ T(a_1)\right],\, \forall a>a_1 \}.
\end{align*}
 Recall that $T(a_1)$ is the stopping time for the CuSum algorithm with $a_1$ as the threshold. Under Assumption \ref{Assumption:FiniteKLDivergence}, using the standard performance analysis technique for the CuSum algorithm (e.g., P.142 of \cite{poor2009quickest}), one sees that $\mathbb{E}_{\infty}\left[ T(a_1)\right]$ is finite for any finite $a_1$. Using the same analysis for $\mathcal{E}(a_1,\epsilon)$ (see the proof of Lemma~\ref{Proposition:ComRateSatisfied}), one sees that $\mathcal{E}'(a_1,\epsilon)$ is not empty for any $a_1$ and $\epsilon$.  Furthermore, when $\epsilon_1$ is small enough, it must belong to both $\mathcal{E}(a_1,\epsilon)$ and $\mathcal{E}'(a_1,\epsilon)$, then $\mathcal{E}^{*}(a_1,\epsilon)$
is a non-empty set for any $a_1$ and $\epsilon$.

We are now ready to present the second main theorem.

\begin{theorem}  \label{Theorem:AsymOpt2Levels}
For any $\epsilon > 0$, when $a=\ln\zeta$, $\epsilon_1\in \mathcal{E}^{*}(a_1,\epsilon)$ with any $0<a_1<a$ are used, the CuSum-AC algorithm satisfies the ARFLA constraint~\eqref{Equantion:ProCons1} and communication rate constraint~\eqref{Equantion:ProCons2}. Furthermore, the CuSum-AC algorithm is asymptotically ($\zeta\to\infty$) optimal for Problem~\ref{Problem:Lorden}, i.e., as $\zeta\to\infty$,
\begin{align*}
d_L(T_c(\ln\zeta),\Psi) = \frac{\ln \zeta}{\mathbf{I}(f_1||f_0)}(  1+ \smallO 1 ).
\end{align*}

\end{theorem}
\begin{remark} \label{Remark:TotalCost}
The globally (for any communication rate constraint $\epsilon > 0$) asymptotic optimality of the Cusum-AC algorithm stated in Theorem~\ref{Theorem:AsymOpt2Levels} relies on the following two factors: (i) the asymmetric behavior of the detection statistic for the CuSum-AC algorithm $s_k$ on pre-change and post-change hypotheses, (ii) the fact that the communication rate defined in~\eqref{Equation:CommunicationConstraint} is merely on the pre-change hypothesis. On the one hand, under $\mathbb{P}_{\infty}$, the expected duration of $s_k$ being above $a_1$ is  finite even when $a$ is infinite (i.e., $T_{a_1}^{\infty}$ defined in~\eqref{Equation:DurationAbove} is finite), then one can choose $\epsilon_1$ (equivalent to $\psi^{*}(\epsilon_1)$) to make the communication rate (ARLFA) of the CuSum-AC algorithm arbitrarily small (large). On the other hand, under $\mathbb{P}_{1}$, the expected duration of $s_k$ being below $a_1$ is finite for any $\epsilon_1>0$, while the expected duration of $s_k$ being above $a_1$ goes to infinity when $a\to\infty$.
\end{remark}




\begin{remark} \label{Remark:Nlevels}
Note that the general $N>2$ levels cases provide more degrees of design freedom (including $a_{N-1},\ldots,a_1$ and $\epsilon_{N-1},\ldots,\epsilon_{1}$) than the two levels case (including $a_1$ and $\epsilon_1$). Following exactly the  same arguments for Theorem~\ref{Theorem:AsymOpt2Levels}, one concludes that the CuSum-AC algorithm with $N>2$ levels is asymptotically ($\zeta\to\infty$) optimal. Though the asymptotic optimality can be achieved for any censoring levels, better detection performance should be expected when $N$ increases if $\zeta$ takes moderate values.
\end{remark}

\begin{remark} \label{Remark:feedbackTimes}
We remark that the feedback transmissions needed for the CuSum algorithm are in general quite few. In particular, the average number of feedback transmission under $\mathbb{P}_{\infty}$ is $\mathbb{E}_{\infty}[N_F]=\frac{2}{1-\mathbb{P}_{\infty}\{\hat{s}_{\eta(0)} < a_1\}}$.
From~\eqref{Equation:RandomStartTVCRepresentation}, the ratio of the average feedback transmission times to the ARLFA is $\frac{\mathbb{E}_{\infty}[N_F]}{\mathbb{E}_{\infty}[T_c(a)]}\leq\frac{2}{\mathbb{E}_{\infty}\left[  \eta(0) \right]
  + \mathbb{E}_{\infty}\left[ \phi\left(\hat{s}_{\eta(0)}\right) |\hat{s}_{\eta(0)} < a_1 \right]\mathbb{P}_{\infty}\{\hat{s}_{\eta(0)} < a_1\}}$, which is usually negligible since the term $\mathbb{E}_{\infty}\left[ \phi\left(\hat{s}_{\eta(0)}\right) |\hat{s}_{\eta(0)} < a_1 \right]\mathbb{P}_{\infty}\{\hat{s}_{\eta(0)} < a_1\}$ is in general quite large.
\end{remark}

\section{Extension to Sensor Networks} \label{Section:ExtensiontoMultipleSensors}
In this section, first we modify the considered problem to sensor networks. Then we generalize the CuSum-AC algorithm presented in Section \ref{Section:Algorithm} and the results obtained in Section \ref{Section:MainResults} to this case.
\subsection{Problem Formulation}
The system is illustrated in Fig. \ref{Fig:blockdiagramnetwork}. Let $\mathcal{M}=\{1,\ldots,M\}$.
As in the one sensor case, it assumed that at sensor~$m$, $\{X_{\{m,1\}},\ldots,X_{\{m,\nu-1\}}\}$ are i.i.d. with pdf $f_{\{0,m\}}$ and $\{X_{\{m,\nu\}},\ldots\}$
are i.i.d. with pdf $f_{\{1,m\}}$.
As in \cite{Veeravalli2001, tartakovsky2008asymptotically}, it is assumed that the change event affects all the sensors simultaneously at $\nu$ and the observations are independent across the sensors, conditioned on the change point.

Like $\gamma_k$ in \eqref{Equation:IndicatorOneSensor}, let $\gamma_{\{m,k\}}$ be indicator whether or not the sensor $m$ sends its observation $X_{\{m,k\}}$ to the fusion center. Let $\psi_{\{m,k\}}$ be the censoring strategy used at the sensor $m$ at the time instant $k$, i.e.,
\[\gamma_{\{m,k\}} = \psi_{\{m,k\}}(X_{\{m,k\}}).\]
Let $\psi^{\mathcal{M}}_k = \{\psi_{\{1,k\}},\ldots,\psi_{\{M,k\}}\}$ be the censoring strategies used at all the sensor nodes at time $k$ and $\Psi^{\mathcal{M}}$ be the censoring policy along the horizon, i.e., $\Psi^{\mathcal{M}}=\{\psi^{\mathcal{M}}_1,\ldots,\psi^{\mathcal{M}}_T\}$.

The average communication rate before the change event happens for the network is defined by
\begin{align} \label{Equation:ComRateConstNetwork}
    r(\Psi^{\mathcal{M}})=\limsup_{n\to\infty}\frac{1}{nM}\mathbb{E}_{\infty}\left[\sum_{k=1}^{n}\sum_{m=1}^{M}\gamma_{\{m,k\}} \big | T \geq n\right].
\end{align}
In \cite{banerjee2014data}, the authors posed communication rate constraint for each channel in the multi-channel setting (the affected subset of the sensors is unknown). Since the change event affects all the sensors simultaneously in our case, we instead use the average communication rate of the whole network \eqref{Equation:ComRateConstNetwork}. Given $T$ and $\Psi^{\mathcal{M}}$, the ARLFA is defined in the same way as in the one sensor case, i.e.,
\begin{align*}
    g(T,\Psi^{\mathcal{M}})=\mathbb{E}_{\infty}\left[ T \right].
\end{align*}
Let
\begin{align*}
    \mathcal{I}^{m}_k=\big\{&(1,\gamma_{\{m,1\}},\gamma_{\{m,1\}}X_{\{m,1\}}),\ldots,\\
    &\qquad \qquad (k,\gamma_{\{m,k\}},\gamma_{\{m,k\}}X_{\{m,k\}})\big\}
\end{align*}
and $\mathcal{I}^{\mathcal{M}}_k =\{\mathcal{I}^{1}_k,\ldots,\mathcal{I}^{M}_k\}$. Then the Lorden's  detection delay is defined by
\begin{align*}
    d_{L}(T,\Psi^{\mathcal{M}}) = \sup_{1\leq\nu<\infty} \big\{ \esssup_{\mathcal{I}^{\mathcal{M}}_{\nu-1}} \mathbb{E}_{\nu}[(T-\nu+1)^+|\mathcal{I}^{\mathcal{M}}_{\nu-1}] \big\}.
\end{align*}
Then the problem we are interested in is as follow:
\begin{problem}  \label{Problem:LordenNetwork}
    \begin{align*}
        \minimize_{T,\Psi^{\mathcal{M}}}& \qquad d_{L}(T,\Psi^{\mathcal{M}}),\\
        \subjectto& \qquad g(T,\Psi^{\mathcal{M}}) \geq \zeta, \addtag \label{Equantion:ProCons1Network} \\
                  & \qquad r(\Psi^{\mathcal{M}})\leq \epsilon, \addtag \label{Equantion:ProCons2Network}
    \end{align*}
where $\zeta\geq1$ is a given lower bound of the ARLFA.
\end{problem}

\subsection{CuSum-AC Algorithm for Multiple Sensors Case}
We only present and focus on the CuSum-AC algorithm with $N=2$ levels for sensor networks; $N>2$ levels can be generalized as in the one sensor case.
Let $a$ and $a_1$ be two thresholds. The stopping time is computed as
\begin{align}
    T_{c}^{\mathcal{M}}(a)=\inf\{k:s_k^{\mathcal{M}}\geq a\},
\end{align}
where the detection statistic $s_k^{\mathcal{M}}$ is updated as follows:
\begin{align*}
    \tilde{s}_k^{\mathcal{M}}=&\max\{0,s_{k-1}^{\mathcal{M}}+\sum_{m=1}^{M}\ell^{\psi_{\{m,k\}}}(\gamma_{\{m,k\}},X_{\{m,k\}})\}, \\
    s_k^{\mathcal{M}} = &\left\{
        \begin{array}{ll}
             a_1, & \text{if}\quad s_{k-1}^{\mathcal{M}}<a_1\:\text{and}\: \tilde{s}_k^{\mathcal{M}}\geq a_1,\\
            \tilde{s}_k^{\mathcal{M}}, & \text{otherwise},
        \end{array}
    \right.  \\
    s_0^{\mathcal{M}} = & 0,
\end{align*}
and the censoring strategies are given by
\begin{align}
\psi_{\{m,k\}} = &\left\{
        \begin{array}{ll}
             \psi^{*}(1), & \text{if}\quad s_{k-1}^{\mathcal{M}} \geq a_1,\\
            \psi^{*}(\epsilon_{\{m,1\}}), & \text{otherwise},\\
        \end{array}
    \right.
\end{align}
with $0<\epsilon_{\{m,1\}}\leq1.$
Note that the censoring strategies used at all the sensor nodes are switched simultaneously, which are adaptive to the detection statistic available at the fusion center. This helps reduce the times of feedback and the feedback message can be broadcasted to all the sensor nodes by the fusion center.

Let $X^{\mathcal{M}}_k = \{X_{\{1,k\}},\ldots,X_{\{M,k\}}\}$ and $\epsilon^{\mathcal{M}}_1 = \{\epsilon_{\{1,1\}},\ldots,\epsilon_{\{M,1\}}\}$. Note that $X^{\mathcal{M}}_k$ and $\psi^{\mathcal{M}}_k$ can be regarded as the ``vector version" of $X_k$ and $\psi_k$ in one sensor case, respectively. The CuSum-AC algorithm for the multiple sensors case is equivalent to its counterpart in one sensor case working with $X^{\mathcal{M}}_k$ and $\psi^{\mathcal{M}}_k$. Thus the following theorems are straightforward.
\begin{theorem} \label{Theorem:EqualizerNetwork}
The stopping time $T_{c}^{\mathcal{M}}(a)$ is an equalizer rule for Problem \ref{Problem:LordenNetwork}, i.e., for any $\nu\geq 1$,
\begin{align*}
d_{L}(T_{c}^{\mathcal{M}}(a),\Psi^{\mathcal{M}}) = \esssup \mathbb{E}_{\nu}[(T_{c}^{\mathcal{M}}(a)-\nu+1)^+|\mathcal{I}^{\mathcal{M}}_{\nu-1}].
\end{align*}
\end{theorem}

Define $\mathscr{E}_{*}^{\mathcal{M}}(\epsilon,a_1)$ as the counterpart of $\mathcal{E}^{*}(a_1,\epsilon)$. To see that $\mathscr{E}_{*}^{\mathcal{M}}(\epsilon)$ is not empty, one can pose an additional constraint that $\epsilon_{\{1,1\}}=\cdots=\epsilon_{\{M,1\}}$. Then the arguments follow straightforwardly from that of $\mathcal{E}^{*}(a_1,\epsilon)$.

\begin{theorem}  \label{Theorem:AsymOptNlevelsNetwork}
For any $\epsilon >0$, when $a=\ln\zeta$, $\epsilon^{\mathcal{M}}_1\in\mathscr{E}_{*}^{\mathcal{M}}(\epsilon,a_1)$ with any $0<a_1<a$ are used, the CuSum-AC algorithm satisfies the ARFLA constraint~\eqref{Equantion:ProCons1Network} and communication rate constraint~\eqref{Equantion:ProCons2Network}. Furthermore, the CuSum-AC algorithm is asymptotically optimal for Problem~\ref{Problem:LordenNetwork}, i.e., as $\zeta\to\infty$,
\begin{align*}
d_L(T_{c}^{\mathcal{M}}(\ln\zeta),\Psi^{\mathcal{M}}) = \frac{\ln \zeta}{\mathbf{I}(f_{\{1,m\}}||f_{\{0,m\}})}(  1+ \smallO 1 ).
\end{align*}
\end{theorem}

\section{Numerical Examples} \label{Section:Numerical Example}
For simulations, we consider the problem of mean shift detection in Gaussian noise. It is assumed that $M=3$ identical sensors are deployed and the pre-change and post-change distributions are $f_{\{0,m\}}\thicksim\mathcal{N}(0,1)$ and $f_{\{1,m\}}\thicksim\mathcal{N}(0.5,1),$ respectively.
For simplicity, in each example, the sensors use an identical censoring strategy.

%

\emph{ Example 1.}   The asymptotic optimality of the CuSum-AC algorithm is examined. We compare the detection performance of the CuSum-AC algorithm with that of the CuSum algorithm (which is the optimal one when there is no communication rate constraint). With different ARLFA's (sufficiently large), the detection delays (i.e., $\mathbb{E}_{1} \left[  T\right]$) of these algorithms  are simulated.
For the CuSum-AC algorithm, two communication rate constraints are considered, i.e., $\epsilon = 0.7$ or $\epsilon = 0.4$.
Note that given a communication rate and ARLFA constraint, there may exist multiple admissible combinations of the parameters (i.e., $a,a_1,\epsilon_1$). To alleviate the computation burden, we set $\epsilon_1 = 0.63$ and $a_1 = 0.78$ for the case $\epsilon =0.7$ and $\epsilon_1 = 0.27$ and $a_1 = 0.79$ for the case $\epsilon =0.4$. The value of the threshold $a$ is determined to make the communication rate constraint to be satisfied equally. Since given $a_1$ and $\epsilon_1$, the communication rate is \emph{not} strictly monotonic with $a$, multiple $a$'s (which have different ARLFA's) can be found. In fact, given $a_1$ and $\epsilon_1$, the communication rate remains the same when $a$ varies if $a$ is sufficiently large. The simulation results are given in Fig. \ref{Fig:asymptotic}. It can be seen that as the ARLFA increases, the difference between the delay of the CuSum-AC algorithm (with communication rate either $\epsilon =0.7$ or $\epsilon = 0.4$) remains approximately constant. This verifies the asymptotic optimality, since the difference will be negligible when the ARLFA goes to infinity. Furthermore, we can see that for the CuSum-AC algorithm with the same ARLFA, the one which has the smaller communication rate (i.e., $\epsilon =0.4$) has the larger detection delay. This is consistent with our intuition that better detection performance can be expected when more communication resources are used. We also note that when the communication rate is $0.4$, the delay of the CuSum-AC algorithm is only around $1.2$ time slots larger than that of the CuSum algorithm (the communication rate of which is 1). This shows good detection delay versus communication rate trade-off for the CuSum-AC algorithm, which will be shown further in the next example.

\begin{figure}
  \centering
  \includegraphics[scale=0.45]{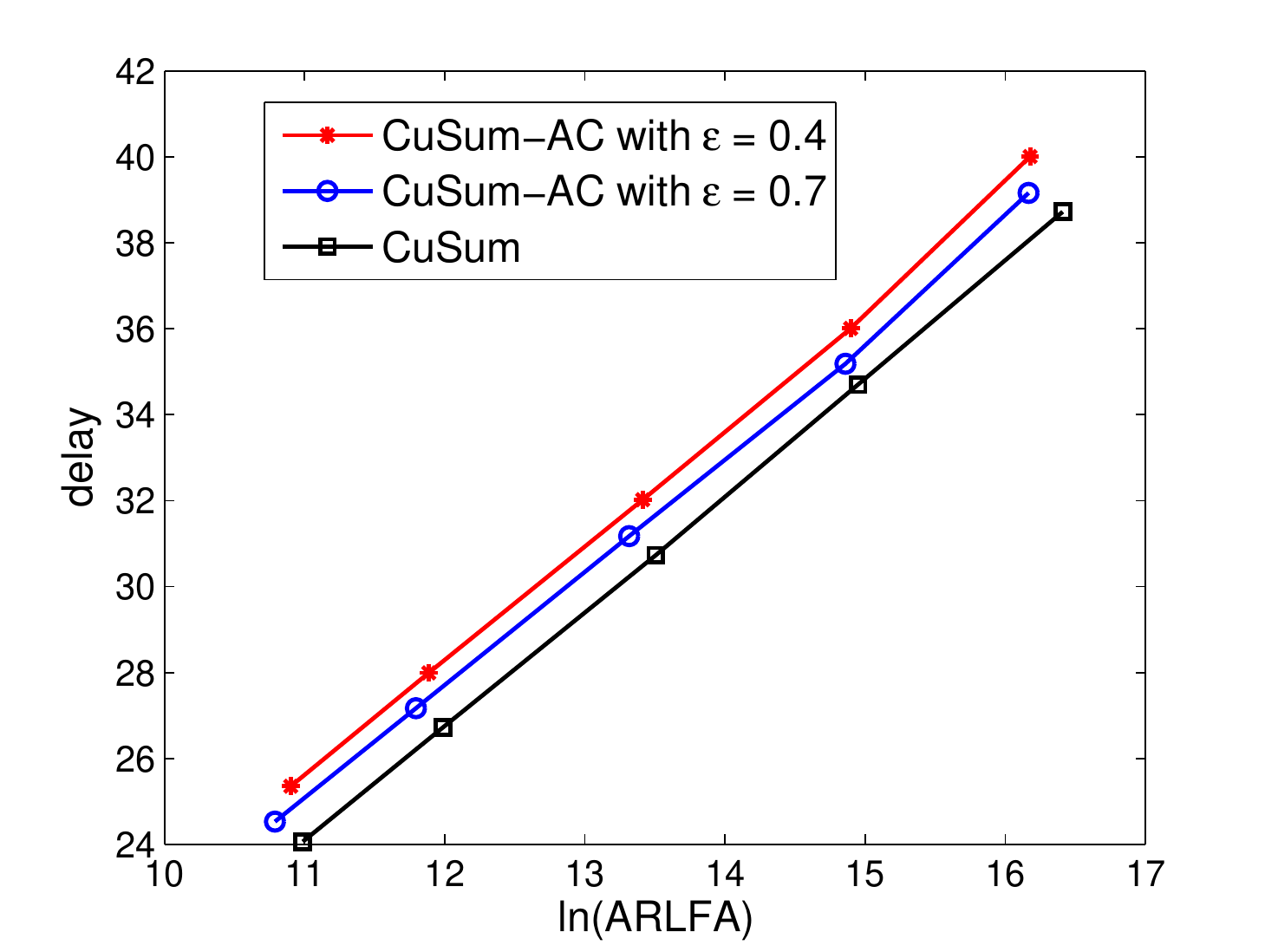}\vspace{-1mm}
  \caption{ Detection delay versus the ARLFA for the CuSum algorithm and the CuSum-AC algorithm.} \label{Fig:asymptotic}
  \vspace{-4mm}
\end{figure}

\emph{ Example 2.} We compare the CuSum-AC algorithm with the CuSum algorithm with the random transmission policy. In the random transmission policy, whether an observation is transmitted or not is determined randomly such that the communication rate is satisfied. The random transmission policy is quite simple and serves as a lower bound of the detection performance in some sense.  We plot Fig.~\ref{Fig:tradeoff2} in the following way. The ARLFA is fixed to be $10\,000$, i.e., the parameters for the algorithm ($a$ for the random transmission control policy, and $a$, $a_1$ and $\epsilon_1$ for the CuSum-AC algorithm) should be chosen such that the associated ARLFA is around $10\,000$. The parameters for the CuSum-AC algorithm are determined using the brute-force search technique.
Multiple admissible combinations of the parameters $a$, $a_1$ and $\epsilon_1$ exist, among which the one that has the smallest detection delay is used. As depicted in Fig.~\ref{Fig:tradeoff2},
the CuSum-AC algorithm significantly outperforms the CuSum algorithm with the random transmission policy, in particular when the allowed communication rate is small.
 One also should note that the CuSum-AC algorithm has quite nice performance \textit{per se}. In particular, when $\epsilon=0.1$, the detection delay of the CuSum-AC algorithm is only $7$ time slots larger than that of the CuSum algorithm (when there is no communication rate constraint, the CuSum-AC algorithm reduces to the CuSum algorithm).

\begin{figure}
  \centering
  \includegraphics[scale=0.45]{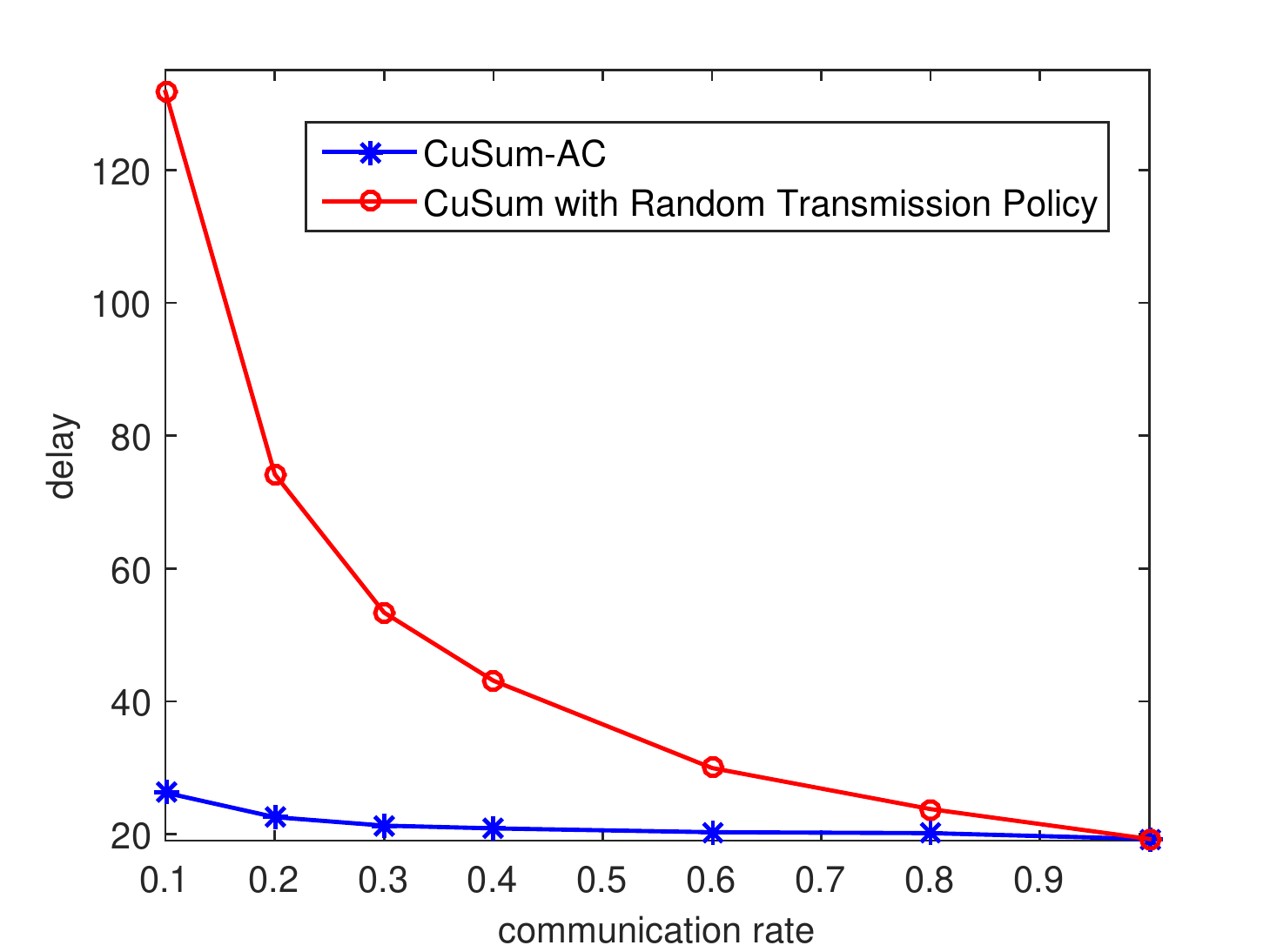}
  \caption{ Detection delay versus the communication rate for the CuSum-AC algorithm  and the CuSum algorithm with random transmission policy. } \label{Fig:tradeoff2}
  \vspace{-4mm}
\end{figure}

\section{Conclusion and Future Work} \label{Section:Conclusion}
In this paper, we have studied the problem of minimax quickest change detection with communication rate constraints.
The constraint is posed by limited energy at the remote sensors.
An extension of the classical Lorden's formulation was studied.
We proposed the CuSum-AC algorithm: the CuSum algorithm is used at the fusion center and adaptive censoring strategies are used at the sensor nodes.
The CuSum-AC algorithm was proved to be an equalizer rule, and be globally asymptotically optimal for any positive communication rate constraint, as the ARLFA goes to infinity.
The numerical simulations showed that the CuSum-AC algorithm has a better detection performance versus communication rate trade-off than the CuSum algorithm with random transmission control policy.

For the future work, there are two interesting directions. One is to explore the relationship between the detection performance when the ARLFA takes moderate values and the censoring strategy being used (in particular, to find whether there exists strictly optimal censoring strategy for every possible ARLFA). The other one is to study the problem in the multi-channel setting (the change event only affects a subset of the sensors).

\section*{Appendix A \\ Proof of Theorem \ref{Theorem:Equalizer}  }
Before proceeding, we first give a supporting lemma as follows. Recall that the stopping time $T_{c}(z,a)$ can be viewed as the first time that the CuSum-AC algorithm reaches the threshold $a$, when starting at $z$. About $T_{c}(z,a)$, the follows hold.
\begin{lemma}  \label{Lemma:EualizerRule}
For any $a$ and $0\leq z < a$,
    \begin{align}
    \mathbb{E}_{1} \left[ T_{c}(z,a) \right] \leq \mathbb{E}_{1} \left[ T_{c}(a) \right]. \label{Equation:EquaLemma}
    \end{align}
\end{lemma}
\begin{proof}
The proof is done by cases.

\emph{Case $0\leq z< a_1$:}
Because of the reset action of the CuSum-AC algorithm when it crosses $a_1$ from below, for $0\leq z< a_1$, we have
\begin{align}
    \mathbb{E}_{1} \left[ T_c(z,a) \right]
    =\mathbb{E}_{1} \left[ \phi(z) \right] + \mathbb{E}_{1} \left[ T_c(a_1,a) \right].  \label{Equation:EuqLemmaSmallZParti}
\end{align}
The quantity $\mathbb{E}_{1} \left[ T_c(a_1,a) \right]$ is a common term. Thus to obtain Lemma~\ref{Lemma:EualizerRule}, we only need to prove that for any $0\leq z<a_1$
\begin{align}
\mathbb{E}_{1} \left[ \phi(z) \right] \leq \mathbb{E}_{1} \left[ \phi(0) \right], \label{Equation:EqualizerLemmaCuSumPathwise}
\end{align}
which is easily obtained by sample path arguments.


Case $a_1\leq z < a$: Starting at $z$, there are two possible ways for the CuSum-AC algorithm to reach the threshold $a$ eventually. One is that the algorithm never returns to the censored region before it stops, i.e., $s_k\geq a_1 $ along the whole horizon; we denote this event by $\mathcal{Z}^{\uparrow}_z$. The other one is that the detection statistic $s_k$ once crosses $a_1$ from up before the algorithm stops, which is denoted by $\mathcal{Z}^{\downarrow}_z$.
Let $p(z)=\mathbb{P}_{1}\left\{  \mathcal{Z}^{\uparrow}_z  \right\}.$
We then have
\begin{align*}
     &\mathbb{E}_{1} \left[ T_c(z,a) \right] \\
   =\, & p(z)\,\mathbb{E}_{1} \left[ T_c(z,a)  \big |  \mathcal{Z}^{\uparrow}_z \right] + \left( 1- p(z)\right)\,\mathbb{E}_{1} \left[T_c(z,a) \big | \mathcal{Z}^{\downarrow}_z \right]  \\
   =\,& p(z)\,\mathbb{E}_{1} \left[ T_c(z,a)  \big |  \mathcal{Z}^{\uparrow}_z \right] \\
   &+ \left( 1- p(z)\right) \left(   \mathbb{E}_{1} \left[\eta(z) \big | \mathcal{Z}^{\downarrow}_z \right]
    + \mathbb{E}_{1} \left[ T_c(x,a) ; \hat{s}_{\eta(z)}=x \big | \mathcal{Z}^{\downarrow}_z \right]    \right) \\
    \defeq \,& p(z) t^{\uparrow}+\left( 1- p(z)\right)\left( t^{\downarrow}_1+ t^{\downarrow}_2\right).
\end{align*}
The physical meaning of $t^{\downarrow}_1$ is the conditional average time it takes for the CuSum-AC algorithm to cross $a_1$ from up for the first time, when starting at $z$.

Before proceeding, we first define a stopping time $\tilde{T}(a)$ as follows. This is a stopping time for a variant of the CuSum algorithm that works in the same manner with the CuSum algorithm but starts at $a_1$ and bounded below by $a_1$:
\begin{align*}
    \tilde{T}(a)=\inf\{k:\tilde{c}_k\geq a\},  
\end{align*}
where $\tilde{c}_k$ involves by
\begin{align*}
    \tilde{c}_k&=\max \left( \tilde{c}_{k-1}+\ell(X_k), a_1 \right),\\
    \tilde{c}_0&=a_1.
\end{align*}
For the quantities $p(z), t^{\uparrow} $ and $t^{\downarrow}_1$, from Lemma $5$ of \cite{banerjee2012data}, we then have an inequality: for any $t^{*} \geq \mathbb{E}_{1}[ \tilde{T}(a)]$,
\begin{align}
p(z) t^{\uparrow}+\left( 1- p(z)\right)\left( t^{\downarrow}_1+ t^{*}\right) \leq t^{*}.  \label{Equation:PzT1Relation}
\end{align}
The reset action for the CuSum-AC algorithm when crossing $a_1$ from below together with the fact that $\{X_k\}'s$ are i.i.d. under the measure $\mathbb{P}_{1}$ yields
\begin{align*}
 \mathbb{E}_{1} \left[ T_c(a_1,a) \right] \geq \mathbb{E}_{1}[ \tilde{T}(a)].
\end{align*}
Combining~\eqref{Equation:EuqLemmaSmallZParti}, one obtains
\begin{align}
 \mathbb{E}_{1} \left[ T_c(a)\right] \geq \mathbb{E}_{1}[ \tilde{T}(a)]. \label{Equation:TVCLargerCuSumVariant}
\end{align}

We then study the quantity $t^{\downarrow}_2$.
\begin{align*}
t^{\downarrow}_2= & \mathbb{E}_{1} \left[ T_c(x,a) ; \hat{s}_{\eta(z)}=x \big | \mathcal{Z}^{\downarrow}_z \right] \\
=& \int_{0\leq x < a_1} \mathbb{E}_{1} \left[ T_c(x,a) \big | \mathcal{Z}^{\downarrow}_z; \hat{s}_{\eta(z)}=x \right] \mathrm{d}\mathbb{P}_{1}\{ \hat{s}_{\eta(z)} \leq x \big | \mathcal{Z}^{\downarrow}_z \}  \\
\stackrel{(e1)}{=}& \int_{0\leq x < a_1} \mathbb{E}_{1} \left[ T_c(x,a) \right] \mathrm{d}\mathbb{P}_{1}\{ \hat{s}_{\eta(z)} \leq x \big | \mathcal{Z}^{\downarrow}_z \} \\
\stackrel{(ie1)}{\leq}& \esssup_{0\leq x < a_1} \mathbb{E}_{1} \left[ T_c(x,a) \right] \\
\leq & \mathbb{E}_{1} \left[ T_c(a) \right].  \addtag  \label{Equation:TVCLargerT2}
\end{align*}
The equality $(e1)$ follows from the Markovity of the detection statistic for the CuSum-AC algorithm. Thus, given $\hat{s}_{\eta(z)}=x$, from the time instant $k=\eta(z)$ on, the evolution of the CuSum-AC algorithm, which starts at $z$, is exactly the same with that of a new CuSum-AC algorithm that starts at $x$. The inequality $(ie1)$ holds because of H\"{o}lder's inequality. 

Then for any $a_1\leq z < a$,
\begin{align*}
\mathbb{E}_{1} \left[ T_c(z,a)\right]
=& p(z) t^{\uparrow}+\left( 1- p(z)\right)\left( t^{\downarrow}_1+ t^{\downarrow}_2\right) \\
\stackrel{(ie1)}{\leq}& p(z) t^{\uparrow}+\left( 1- p(z)\right)\left( t^{\downarrow}_1+ \mathbb{E}_{1} \left[ T_c(a) \right] \right) \\
\stackrel{(ie2)}{\leq}& \mathbb{E}_{1} \left[ T_c(a) \right],
\end{align*}
where the inequality $(ie1)$ follows from \eqref{Equation:TVCLargerT2} and $(ie2)$ follows from \eqref{Equation:PzT1Relation} and \eqref{Equation:TVCLargerCuSumVariant}.
The proof thus is completed.
\end{proof}

We are ready to prove Theorem~\ref{Theorem:Equalizer}.
From the Markovity of the detection statistic $s_k$ for the CuSum-AC algorithm, one can see that the average detection delay is measurable with respect to $s_{\nu-1}$, i.e.,
\begin{align*}
\mathbb{E}_{\nu}[(T_c(a)-\nu+1)^+|\mathcal{I}_{\nu-1}] = &\mathbb{E}_{\nu}[(T_c(a)-\nu+1)^+|s_{\nu-1}]\\
=& \mathbb{E}_{\nu}[(T_c(s_{\nu-1},a)].
 \end{align*}
Note that under Assumption \ref{Assumption:FiniteKLDivergence}, for any censoring strategy $\psi^{*}(\epsilon_1)$, from the positive definiteness of K--L divergence \cite{cover2006elements}, one has
$\mathbb{E}_{\infty}\left[\ell^{\psi^{*}(\epsilon_1)}(\gamma_k,X_k)\right] < 0.$
It then follows that
\begin{align*}
\mathbb{P}_{\infty}\{\ell^{\psi^{*}(\epsilon_1)}(\gamma_k,X_k)<0\}\defeq p >0.
\end{align*}
and for any $\nu\geq 1$
\begin{align}
\mathbb{P}_{\infty}\{ s_{\nu-1}=0  \} \geq p ^{\nu-1} >0.  \label{Equation:PositiveS-1Equal0}
\end{align}
We then obtain
\begin{align*}
&\esssup_{\mathcal{I}_{\nu-1}} \mathbb{E}_{\nu}[(T_c(a)-\nu+1)^+|\mathcal{I}_{\nu-1}]\\
& =  \esssup_{s_{\nu-1}} \mathbb{E}_{\nu}[T_c(s_{\nu-1},a)] \\
& = \mathbb{E}_{\nu}[T_c(s_{\nu-1}=0,a)],
\end{align*}
where the last equality follows from Lemma \ref{Lemma:EualizerRule} and \eqref{Equation:PositiveS-1Equal0}. The proof thus is completed.
\section*{Appendix B \\ Proof of Lemma~\ref{Proposition:ComRateSatisfied}  }
We first give the upper bound of $r(\Psi)$ when $a, a_1$ and $\psi^{*}(\epsilon_1)$ are known. We then show that by choosing certain $\epsilon_1$ (equivalently with $\psi^{*}(\epsilon_1)$), which is independent of $a$, the upper bound can be any admissible value. The technique used is to interpret the CuSum-AC algorithm as a sequence of two-sided SPRTs as in Section~\ref{section:supportingDefinition}.

Define a random sequence $\{T_i\}$ as
\begin{align*}
T_0=& \phi(0),  \\
T_i=& T_{i-1}+ W_i,  \quad \forall i\geq1,
\end{align*}
where $W_i \defeq  W_i^{(1)}+W_i^{(2)}$. The random variables $W_i^{(1)}$ and $W_i^{(2)}$ are i.i.d. distributed with mean equal to $\mathbb{E}_{\infty}\left[ \eta(0) |\hat{s}_{\eta(0)} < a_1 \right]$ and  $\mathbb{E}_{\infty}\left[ \phi\left(\hat{s}_{\eta(0)}\right) |\hat{s}_{\eta(0)} < a_1 \right]$, respectively. Note that the distribution of $W_i$ is \emph{different} from that of $\Phi$ defined in \eqref{Equation:PhiDistribution}. The stopping time $\Phi$ is for the evolution of the CuSum-AC algorithm (which may stop at some time, i.e., $\hat{s}_{\eta(0)} \geq a$), while for the definition of the communication rate constraint \eqref{Equation:CommunicationConstraint}, we implicitly assume that the CuSum-AC algorithm never stops.

Based on $\{T_i\}$, we define a reward sequence $\{R_i\}$ as
\begin{align*}
R_0=& 0,  \\
R_i=& R_{i-1}+ W_i^{(1)}+J_i,  \quad \forall i\geq1,
\end{align*}
where $J_i \sim B(W_i^{(2)},\epsilon_1)$ is a binomial distributed random variable.

Given $\epsilon_1,a_1$ and $a$,
\begin{align*}
    r(\Psi)&\stackrel{(ie1)}{\leq} \limsup_{n\to\infty}\frac{1}{n-T_0}\mathbb{E}_{\infty}\left[\sum_{k=T_0+1}^{n}\gamma_k \bigg | T \geq n\right]\\
   &\: =  \lim_{i\to\infty} \frac{R_i}{T_i-T_0}\\
    & \stackrel{(e1)}{=} \frac{\mathbb{E}_{\infty}\left[W_1^{(1)}+J_1\right]}{\mathbb{E}_{\infty}\left[W_1\right]}  \\
    & \:= \frac{ \mathbb{E}_{\infty}\left[ \eta(0) |\hat{s}_{\eta(0)} < a_1 \right]+\epsilon_1 \mathbb{E}_{\infty}\left[ \phi\left(\hat{s}_{\eta(0)}\right) |\hat{s}_{\eta(0)}<a_1\right]}    { \mathbb{E}_{\infty}\left[ \eta(0) |\hat{s}_{\eta(0)} < a_1 \right]+ \mathbb{E}_{\infty}\left[ \phi\left(\hat{s}_{\eta(0)}\right) |\hat{s}_{\eta(0)} < a_1 \right]}\\
    & \defeq \bar{r}(\Psi),
\end{align*}
where the inequality $(ie1)$ holds because by the definition of $\phi(0)$, before the time instant $T_0$, all the observations are censored using the censoring strategy $\psi^{*}(\epsilon_1)$, the communication rate of which is $\epsilon_1$ (the lower bound of $r(\Psi)$); the equality $(e1)$ follows from the alternating renewal process theory (Page 173,~\cite{asmussen2003applied}).

Using the sample path arguments, one sees that, given $a_1$, $\mathbb{E}_{\infty}\left[ \eta(0) |\hat{s}_{\eta(0)} < a_1 \right]$ is monotonically nondecreasing with~$a$. Let
\begin{align}  \label{Equation:DurationAbove}
T_{a_1}^{\infty}\defeq\lim_{a\to+\infty} \mathbb{E}_{\infty}\left[ \eta(0) |\hat{s}_{\eta(0)} < a_1 \right].
\end{align}
Since $\mathbb{E}_{\infty}\left[\ell(X_k)\right]<0$ (by Assumption \ref{Assumption:FiniteKLDivergence}),
Corollary $2.4$ of \cite{woodroofe1982nonlinear} yields that $T_{a_1}^{\infty}$ is finite.
Note that $\bar{r}(\Psi)$ is monotonically nondecreasing with $\mathbb{E}_{\infty}\left[ \eta(0) |\hat{s}_{\eta(0)} < a_1 \right]$, then for any $a$,
\begin{align*}
\bar{r}(\Psi) \leq & \frac{T_{a_1}^{\infty} + \epsilon_1 \mathbb{E}_{\infty}\left[ \phi\left(\hat{s}_{\eta(0)}\right) |\hat{s}_{\eta(0)}<a_1\right]}    { T_{a_1}^{\infty} + \mathbb{E}_{\infty}\left[ \phi\left(\hat{s}_{\eta(0)}\right) |\hat{s}_{\eta(0)} < a_1 \right] }  \\
= & \frac{      \frac{ T_{a_1}^{\infty}}{\mathbb{E}_{\infty}\left[ \phi\left(\hat{s}_{\eta(0)}\right) |\hat{s}_{\eta(0)}<a_1\right]} +\epsilon_1  }    { \frac{T_{a_1}^{\infty}}{\mathbb{E}_{\infty}\left[ \phi\left(\hat{s}_{\eta(0)}\right) |\hat{s}_{\eta(0)} < a_1 \right]} + 1 }  \\
\defeq & \bar{\bar{r}}(\Psi).
\end{align*}
For any censoring strategy $\psi^{*}(\epsilon_1)$, $\mathbb{E}_{\infty}\left[ \ell^{\psi^{*}(\epsilon_1)}(\gamma_k,X_k)\right]\leq 0$. Furthermore, as $\epsilon_1\to 0$, $\mathbb{E}_{\infty}\left[ \ell^{\psi^{*}(\epsilon_1)}(\gamma_k,X_k)\right] \to 0$ \footnote{This approach is not necessarily monotonic. In some cases, $\mathbb{E}_{\infty}\left[ \ell^{\psi^{*}(\psi_1)}(\gamma_k,X_k)\right]$ can be zero for non-zero $\epsilon_1$'s.}. Note that when $\epsilon_1=0$, not only the mean of the log-likelihood ratio but also the random variable $\gamma_kX_k$ reduces to constant $0$. By Corollary $2.6$ of \cite{woodroofe1982nonlinear}, one sees that as $\mathbb{E}_{\infty}\left[ \ell^{\psi^{*}(\epsilon_1)}(\gamma_k,X_k)\right] \to 0$, $\mathbb{E}_{\infty}\left[ \phi\left(z\right)\right] \to \infty$ for any $0\leq z < a_1$. Note also that because of the reset action whenever the CuSum-AC algorithm crosses $a_1$ from below, $T_{a_1}^{\infty}$ is only related with the distribution of $X_k$ under $\mathbb{P}_{\infty}$, which is independent of the censoring strategy $\psi^{*}(\epsilon_1)$ (i.e., $\epsilon_1$). Given any $\epsilon$, one thus can find a non-empty set $\mathcal{E}(a_1,\epsilon)$ that is independent of $a$ (i.e., independent of the distribution of $\{\hat{s}_{\eta(0)}|\hat{s}_{\eta(0)} < a_1\}$) such that for any censoring strategy $\psi^{*}(\epsilon_1)$ with $\epsilon_1\in \mathcal{E}(a_1,\epsilon)$,
$
\bar{\bar{r}}(\Psi)\leq \epsilon.
$
The proof thus is completed.
\section*{Appendix C \\ Proof of Theorem \ref{Theorem:AsymOpt2Levels}  }
By the definition of $\mathcal{E}^{*}(a_1,\epsilon)$, the CuSum algorithm satisfies the communication rate constraint~\eqref{Equantion:ProCons2}.
Note that the CuSum algorithm is strictly optimal for Problem~\ref{Problem:Lorden} when $\epsilon=1$. We prove Theorem~\ref{Theorem:AsymOpt2Levels} by relating the CuSum-AC algorithm to the CuSum algorithm.

\emph{Step $1$}.  We first prove that the ARLFA of the CuSum-AC algorithm that uses the censoring strategy defined in Theorem~\ref{Theorem:AsymOpt2Levels} is \emph{always larger} than that of the CuSum algorithm. This is due to the definition of $\mathcal{E}^{*}(a_1,\epsilon)$ in~\eqref{Equation:OptimalEpsilonSetDefinition}, by which $\epsilon_1$ (equivalent to $\psi^*(\epsilon_1)$) is appropriately chosen. Note that the ARLFA can be made arbitrarily large by adjusting $\epsilon_1$.
To this end, we define a stopping time as follows:
\begin{align*}
    \hat{T}(a)=\inf\{k:\hat{c}_k\geq a\},
\end{align*}
where the detection statistic $\hat{c}_k$ evolves by
\begin{align*}
    \hat{c}'_k&=\left( \hat{c}_{k-1}+\ell(X_k) \right)^{+},\\
    \hat{c}_k & = \left\{
        \begin{array}{ll}
            a_1, & \text{if} \quad \hat{c}'_k\geq a_1\:\text{and}\: \hat{c}_{k-1}<a_1,\\
            0, & \text{if} \quad \hat{c}'_k < a_1\:\text{and}\: \hat{c}_{k-1} \geq a_1,\\
            \hat{c}'_k, & \text{otherwise},
        \end{array}
    \right. \\
    \hat{c}_0&=0.
\end{align*}
The difference between $\hat{c}_k$ and $c_k$ in \eqref{Equation:CuSumStoppingTime} for the CuSum algorithm is that $\hat{c}_k$ is reset to be $a_1$ ($0$) whenever it crosses $a_1$ from below (up). Using the sample path arguments, one can see that for any $a$,
\begin{align}
    \mathbb{E}_{\infty}\left[  \hat{T}(a)  \right] \geq \mathbb{E}_{\infty}\left[  T(a)  \right]. \label{Equation:ResetCuSumBigger}
\end{align}
Let $N(k)$ be the  number of  $\hat{c}_k$ crossing $a_1$ from below by time instant $k$, i.e,
\begin{align*}
N(k) & = \left\{
        \begin{array}{ll}
            N(k-1)+1, & \text{if} \quad \hat{c}_k\geq a_1\:\text{and}\: \hat{c}_{k-1}<a_1,\\
            N(k-1), & \text{otherwise},
        \end{array}
    \right.
\end{align*}
with $N(0)=0$. Wald's identity \cite{siegmund1985sequential} yields
\begin{align*}
&\mathbb{E}_{\infty}\left[ \hat{T}(a)  \right] \\
=& \left( \mathbb{E}_{\infty}\left[  \eta(0) \right]+ T(a_1)\mathbb{P}_{\infty}\{\hat{s}_{\eta(0)} < a_1\}  \right)
\mathbb{E}_{\infty} \left[N\left( \hat{T}(a)\right)  \right].  \addtag  \label{Equation:ResetCuSumRepresentation}
\end{align*}

Let
$
    T^{\diamond}_c(a)=\inf\{k:s^{\diamond}_k\geq a\},
$
where $s^{\diamond}_k$ is evolved in the same manner with $s_k$ in~\eqref{Equation:2LevelsUpdate} for the stopping time $T_c(a)$ except for the starting point. The detection statistic $s_k$ starts at $s_0=0$, while $s^{\diamond}_k$ starts with a random variable $s^{\diamond}_0=x$ and the distribution of $x$ is the same with that of the random variable $\{\eta(0)|\hat{s}_{\eta(0)} < a_1\}$. Using the sample path arguments, one sees that
\begin{align}
    \mathbb{E}_{\infty}\left[  T_c(a)  \right] \geq \mathbb{E}_{\infty}\left[  T^{\diamond}_c(a)  \right].
    \label{Equation:RandomStartTVCSmaller}
\end{align}
Let $N^{\diamond}(k)$ be the number of  $s^{\diamond}_k$ crossing $a_1$ from below by time instant $k$, which is defined in the same manner with $N(k)$. Also by Wald's identity, one has
\begin{align*}
&\mathbb{E}_{\infty}\left[ T^{\diamond}_c(a)   \right] \\
= &\left( \mathbb{E}_{\infty}\left[  \eta(0) \right]
  + \mathbb{E}_{\infty}\left[ \phi\left(\hat{s}_{\eta(0)}\right) |\hat{s}_{\eta(0)} < a_1 \right]\mathbb{P}_{\infty}\{\hat{s}_{\eta(0)} < a_1\}  \right) \\
&\quad\mathbb{E}_{\infty} \left[N^{\diamond} \left( T^{\diamond}_c(a) \right)  \right].  \addtag \label{Equation:RandomStartTVCRepresentation}
\end{align*}
Both $\hat{c}_k$ and $s^{\diamond}_k$ are reset to be $a_1$ when they cross $a_1$ from below, we thus obtain the following:
\begin{align*}
\mathbb{E}_{\infty} \left[N\left( \hat{T}(a)\right)  \right] =& \mathbb{E}_{\infty} \left[N^{\diamond} \left( T^{\diamond}_c(a) \right)  \right]  \addtag \label{Equation:EqualNumnber}\\
=& \frac{1}{1-\mathbb{P}_{\infty}\{\hat{s}_{\eta(0)} < a_1\}},
\end{align*}
where the last equality follows because both $N(k)$ and $N^{\diamond}(k)$ are geometrically distributed.

Combining \eqref{Equation:OptimalEpsilonSetDefinition}, \eqref{Equation:ResetCuSumBigger}, \eqref{Equation:ResetCuSumRepresentation}, \eqref{Equation:RandomStartTVCSmaller}, \eqref{Equation:RandomStartTVCRepresentation} and \eqref{Equation:EqualNumnber}, one can see that for any $\epsilon$, when the censoring strategy $\psi^{*}(\epsilon_1)$ with $\epsilon_1\in \mathcal{E}^{*}(a_1,\epsilon)$ is used,
\begin{align}
\mathbb{E}_{\infty}\left[  T_c(a)  \right] \geq \mathbb{E}_{\infty}\left[  T(a)  \right].  \label{Equation:ARLTVCBigger}
\end{align}
Then by the established performance results of the CuSum algorithm~\cite{poor2009quickest}, the CuSum-AC algorithm satisfies the ARLFA constraint~\eqref{Equantion:ProCons1}.

\emph{Step $2$}. We show that for any $\epsilon$, as $a\to\infty$,
\begin{align*}
\mathbb{E}_{1}\left[  T_c(a)  \right] = \mathbb{E}_{1}\left[  T(a)  \right] \left(  1+ \smallO 1  \right).
\end{align*}
The intuition is as follows. Since for any finite $a_1$ and $\epsilon_1>0$, the duration of $s_k$ staying below $a_1$ (when censored observations are used) is finite, then as $a\to\infty$, the duration of $s_k$ being above $s_k$ (when raw observations) dominates. The asymptotic first-order behavior of the detection delay of the CuSum-AC algorithms thus resembles that of classical CuSum algorithm.

Let
\begin{align*}
T_1 = & \sum_{k=1}^{T_c(a)} \mathbf{1}_{ \{ s_k \geq a_1\} }, \\
T_2 = & T_c(a)-T_1.
\end{align*}
Recall that $s_k$ is the detection statistic for the CuSum-AC algorithm.
The quantity $T_1$ ($T_2$) can be viewed as the duration that $s_k$ stays above (below) $a_1$. Following similar arguments in \emph{Step $1$}, one obtains that
\begin{align*}
\mathbb{E}_{1}[T_1] = & \mathbb{E}_{1}\left[  \eta(0) \right] \frac{1}{1-\mathbb{P}_{1}\{\hat{s}_{\eta(0)} < a_1\}}, \\
\mathbb{E}_{1}[T_2] \leq & \mathbb{E}_{1}\left[ \phi(0) \right]\frac{1}{1-\mathbb{P}_{1}\{\hat{s}_{\eta(0)} < a_1\}}.
\end{align*}

Since $\mathbf{I}(f_1||f_0)>0$, by Corollary $2.4$ of \cite{woodroofe1982nonlinear}, one obtains that as $a\to\infty$,
$\mathbb{E}_{1}\left[  \eta(0) \right] \to \infty$.
Note that for any $\epsilon_1 >0$, $\mathbf{I}^{\psi^{*}(\epsilon_1)}(f_{1}||f_0)>0$. From the established performance analysis technique for the CuSum algorithm (e.g., P. 142 of \cite{poor2009quickest}), one sees that
$\mathbb{E}_{1}\left[ \phi(0) \right] < \infty.$
Note that $\mathbb{E}_{1}\left[ \phi(0) \right]$ is only related to $a_1$ and $\psi^{*}(\epsilon_1)$. By the definition of $\mathcal{E}^{*}(a_1,\epsilon)$, $\psi^{*}(\epsilon_1)$ is independent of $a$. The following thus can be obtained:
\begin{align*}
\frac{\mathbb{E}_{1}[T_2]}{\mathbb{E}_{1}[T_1]} \leq \frac{\mathbb{E}_{1}\left[ \phi(0) \right]}{\mathbb{E}_{1}\left[  \eta(0) \right]} \to 0, \quad \text{as} \: a\to\infty.
\end{align*}
Then $a\to \infty$,
\begin{align}
\mathbb{E}_{1}\left[  T_c(a)  \right] = \mathbb{E}_{1}[T_1] (\left(  1+ \smallO 1  \right)). \label{Equation:AsympE_11}
\end{align}

Because of the reset action when $s_k$ crosses $a_1$ from below, the following holds:
\begin{align*}
\mathbb{E}_{1}[T_1] \leq & \mathbb{E}_{1}\left[  T(a-a_1)  \right] \\
=& \frac{\ln (a-a_1)}{\mathbf{I}(f_1||f_0)}(  1+ \smallO 1 ), \quad \text{as} \: a\to\infty. \addtag \label{Equation:AsympE_12}
\end{align*}
Note that as $a\to\infty$
\begin{align}
\mathbb{E}_{1} \left[ T(a) \right] = \frac{\ln (a)}{\mathbf{I}(f_1||f_0)}(  1+ \smallO 1 ).
\label{Equation:AsympE_13}
\end{align}
Combining \eqref{Equation:AsympE_11}, \eqref{Equation:AsympE_12} and \eqref{Equation:AsympE_13}, one obtains that
for any $\epsilon$, as $a\to\infty$,
\begin{align*}
\mathbb{E}_{1}\left[  T_c(a)  \right] &\leq \mathbb{E}_{1}\left[  T(a)  \right] \left(  1+ \smallO 1  \right) \\
& \stackrel{(e1)}{=} \mathbb{E}_{1}\left[  T(a)  \right] \left(  1+ \smallO 1  \right),
\addtag \label{Equation:DelayAsymptoApproach}
\end{align*}
where $(e1)$ holds because given \eqref{Equation:ARLTVCBigger}, $\mathbb{E}_{1}\left[  T(a)  \right]$ is the lower bound of $\mathbb{E}_{1}\left[  T_c(a)  \right]$ as $a\to\infty$.

\emph{Step $3$.} By Theorem \ref{Theorem:Equalizer},
\begin{align*}
d_L(T_c(a),\Psi) = \mathbb{E}_{1}\left[  T_c(a)  \right].
\end{align*}
Note that $T(a)$ is asymptotically optimal for Problem~\ref{Problem:Lorden}, as $\zeta\to\infty$ (i.e., $a\to\infty$), when $\epsilon=1$. Combining \eqref{Equation:ARLTVCBigger} and \eqref{Equation:DelayAsymptoApproach}, one obtains that the CuSum-AC algorithm is asymptotically optimal. The proof thus is complete.

\bibliographystyle{IEEETran}
\bibliography{xq_reference}

\end{document}